\documentclass[a4paper,11pt]{article}

\usepackage{graphicx}
\usepackage{amsmath}
\usepackage{amssymb}
\usepackage{paralist}
\usepackage{mathtools,mathrsfs}
\usepackage{csquotes}
\usepackage{thm-restate}
\usepackage{tikz}
\usepackage{microtype}
\usepackage[margin=1in]{geometry}

\usepackage{hyperref}
\usepackage[nameinlink]{cleveref}
\hypersetup{
	colorlinks,
	linkcolor={red!50!black},
	citecolor={blue!50!black},
	urlcolor={blue!80!black}
}

\usepackage{amsthm}
\newtheorem{theorem}{Theorem}

\newtheorem{lemma}{Lemma}
\newtheorem{corollary}{Corollary}

\allowdisplaybreaks
\title{Scheduling on a Stochastic Number of Machines}

\author{Moritz {Buchem}\thanks{Technische Universität München, Germany. Email: \texttt{m.buchem@tum.de}.}
\and
Franziska {Eberle}\thanks{{Technische Universität Berlin, Germany. Email: \texttt{f.eberle@tu-berlin.de}. Supported by the Dutch Research Council (NWO),
Netherlands Vidi grant 016.Vidi.189.087.}}
\and
Hugo Kooki {Kasuya Rosado}\thanks{Technische Universität München, Germany. Email: \texttt{hugo.rosado@tum.de}.}
\and
Kevin {Schewior}\thanks{University of Southern Denmark, Odense, Denmark. Email: \texttt{kevs@sdu.dk}. Supported by the Independent Research Fund Denmark, Natural Sciences, grant DFF-0135-00018B.}
\and
Andreas {Wiese}\thanks{Technische Universität München, Germany. Email: \texttt{andreas.wiese@tum.de}.
}}

\date{}

\newcommand{\bags}{\mathcal{B}}
\newcommand{\calL}{\mathcal{L}}
\newcommand{\guess}{\mathcal{M}}
\newcommand{\opt}{\textsc{Opt}}
\newcommand{\optbags}[1][]{\mathcal{B}^*_{#1}}
\newcommand{\guessbags}[1][]{\hat{\mathcal{B}}_{#1}}
\newcommand{\bigO}{\mathcal{O}}
\newcommand{\eps}{\varepsilon}
\renewcommand{\epsilon}{\varepsilon}
\newcommand{\sand}{\textup{sand}}
\newcommand{\kmax}{K}
\newcommand{\noBags}{M_{k+1,\ldots,\kmax}}
\global\long\def\N{\mathbb{N}}%
\global\long\def\ALG{\textsc{Alg}}%

\newcommand{\profit}{\mathrm{profit}}
\newcommand{\calC}{\mathcal{C}}

\usepackage[textsize=tiny,textwidth=3.5cm,color=Red]{todonotes}

\begin{document}
\pagenumbering{gobble} 

\maketitle           

\begin{abstract}
We consider a new
scheduling problem on parallel identical machines in which the number of machines is initially not known, but it follows a given probability distribution.
Only after all jobs are assigned to a given number of bags, the actual number of machines is revealed. Subsequently, the jobs need to be assigned to the machines without splitting the bags. This is the stochastic version of a related problem introduced by Stein and Zhong [SODA 2018, TALG 2020] and it is, for example, motivated by bundling jobs that need to be scheduled by data centers.
We present two PTASs for the stochastic setting, computing job-to-bag assignments that (i) minimize the expected maximum machine load and (ii) maximize the expected minimum machine load (like in the Santa Claus problem), respectively. The former result follows by careful enumeration combined with known PTASs. For the latter result, we introduce an intricate dynamic program that we apply to a suitably rounded instance.
\end{abstract}

\newpage
\pagenumbering{arabic}

\section{Introduction}
Stein and Zhong~\cite{SteinZ20} recently introduced scheduling problems in which the number of the given
(identical) machines is initially unknown. Specifically, all jobs must be assigned to 
a given number of bags before the actual number of machines is revealed. When that happens, the bags cannot be split anymore and they have to be assigned to the machines as whole bags, optimizing some objective function. Such problems arise, e.g., when ``bundling'' jobs to be scheduled in data centers, where the number of available machines depends on external factors such as momentary demand~\cite{EberleGMMZ23,SteinZ20}.

The aforementioned work (as well as follow-up works~\cite{BalkanskiOSW22,EberleHMNSS23}) focused on the robustness of a job-to-bag assignment. Specifically, they assumed a worst-case number of machines and compared their solution with the in-hindsight optimum for the respective objective function, i.e., a direct job-to-machine assignment without bags.
In contrast to this information-theoretic question, we assume that a distribution of the number of machines is known (e.g., from historical data) and aim to \emph{efficiently compute} a job-to-bag assignment that optimizes the objective function \emph{in expectation}---a common formulation of 
the objective function for stochastic (scheduling) problems~\cite{EberleGMMZ23,Gupta0NS21,GuptaMZ23,ImMP15,KleinbergRT00,Sethuraman16,SkutellaSU16}. In other words, we use a ``fairer'' benchmark for our algorithms, allowing us to sidestep the strong lower bounds by~\cite{SteinZ20}. We are the first to study this novel type of scheduling problem, already proposed in~\cite{SteinZ20}. 

We consider two classic objective functions: minimizing the maximum machine load (makespan) and maximizing the minimum machine load (Santa Claus). Both objectives are well-studied in the deterministic setting, the special case of our problem with one-point distributions, i.e., the distributions in which only one event happens with positive probability. These problems are well understood from a classic approximation perspective: both are known to be strongly NP-hard~\cite{GareyJ75}
and both 
admit Polynomial-Time Approximation Schemes (PTASs)~\cite{HochbaumS87,Woeginger97}, i.e., polynomial-time $(1+\eps)$-approximation algorithms for any $\eps>0$. 
In this paper, surprisingly, we recover the same state for the stochastic versions by designing a PTAS in both cases. 
In contrast to the deterministic setting, we require different techniques tailored to each objective function. 
For the makespan minimization objective, our main technical contribution  is the application and analysis of techniques that have previously been used in approximation schemes for \emph{deterministic} scheduling and packing problems. Our approach for the Santa Claus objective is technically much more intriguing and requires the careful set-up of a novel dynamic program (DP) in order to control its size.

Our results are in stark contrast to classic stochastic scheduling problems, where in some cases the currently best known approximation algorithms have distribution-dependent or even linear guarantees~\cite{ImMP15,SkutellaSU16}. Even for better-understood problems such as load balancing of stochastic jobs on deterministic machines, previous approaches~\cite{EberleGMMZ23,Gupta0NS21,KleinbergRT00} rely on concentration bounds which inherently prohibit approximation ratios of $1+\eps$ for arbitrarily small $\eps>0$.
Moreover, PTASs for stochastic load balancing on deterministic machines are only known for identical machines and Poisson distributed jobs~\cite{de2020efficient,ibrahimpur2021minimum}. We hope that our positive results inspire research for other scheduling problems with a stochastic number of machines, even for (in the classic model with jobs with stochastic processing times) notoriously hard objective functions such as expected weighted sum of completion times. 

\subsection{Our Contribution and Techniques}
Our first result is the following.

\begin{theorem}\label{theo:makespan}
    There is a PTAS for the problem of computing the job-to-bag assignment that minimizes the expected maximum machine load.
\end{theorem}

We first guess the bag sizes of the optimal solution up to a factor of $1+\eps$. For each guess, we check whether there is a corresponding assignment of the jobs to the bags (up to a factor of $1+\eps$),
using the PTAS for bin packing with variable sizes~\cite{HochbaumS88}. Among the guesses that fulfill this condition, we can select the (approximately) best guess using the PTAS for makespan minimization~\cite{HochbaumS87}.

For this approach to yield a PTAS, we need to bound the number of guesses by a polynomial (in the input length). First note that it is straightforward to get down to a \emph{quasi-polynomial} number of guesses (and thus a QPTAS). The approach is to disregard jobs of size $(\eps/n)\cdot p_{\max}$ where $p_{\max}$ is the largest processing time; indeed, for any solution, such jobs make up at most an $\eps$-fraction of the objective-function value. The resulting number of possible guesses for a single bag size is then logarithmic in $n$, leading to a quasi-polynomial number of guesses for the (multi-)set of bag sizes. To get a \emph{polynomial} bound (and thus a PTAS), we make the following crucial observation.
Let $C$ be an estimate for the largest \emph{bag} size, up to a constant factor. While bags of size $\bigO(\eps C)$ cannot be disregarded, it is enough to know their \emph{number} rather than approximate size.
Intuitively, when computing a bag-to-machine assignment, these bags are treated like ``sand'', i.e., as infinitesimal jobs of total volume equal to the total volume of those bags. The number of possible (rounded) bag sizes is hence constant, leading to a polynomial number of guesses for the (multi-)set of bag sizes.

Our second result is significantly harder to achieve.

\begin{restatable}{theorem}{thmsanta}
    \label{thm:santa}
    There is a PTAS for the problem of computing the job-to-bag assignment that maximizes the expected minimum machine load.
\end{restatable}

One may be tempted to try a similar approach as for our first result. Even getting a QPTAS is, however, not possible in the same way as one cannot simply disregard jobs of size $(\eps/n)\cdot p_{\max}$. Consider an instance in which the number of machines is deterministically $M$ and in which there are one job of size $1$ and $M-1$ jobs of size $\eps/(2M)$. Here, ignoring the jobs of size $\eps/(2M)$ leads to $M-1$ empty bags, yielding an objective function value of $0$ instead of the optimal value $\eps/(2M)$.

Also, it is no longer true either that for bags of size $\bigO(\eps C)$ (where $C$ is the size of a largest bag) it is enough to know their total number: Consider an instance with optimal objective-function value $\opt$ and add one huge job of size $\opt/\eps^2$ to the set of jobs and one machine to each scenario. 
Clearly, this new instance has maximum bag size $C \geq \opt/\eps^2$ while the optimal objective-function value does not change since this huge job can safely be packed in its private bag and scheduled on its private machine. (However, crucially for the PTAS for makespan minimization, adding this huge job there would change the objective-function value.) 
In this example, the probability that scenarios with optimal objective-function value much smaller than $\eps C$ occur is $1$. 
To obtain a $(1+\eps)$-approximation, however, one still has to compute a $(1+\eps)$-approximation for the original instance, for which the sizes of bags of size $\bigO(\eps C)$ \emph{are} relevant. 
Of course, the issue with this particular instance could be avoided by removing the huge job in a pre-processing step. However, by concatenating the above original instance at super-constantly many different scales, one can create a new instance where one essentially has to identify the ``relevant scales'' in a preprocessing step.

In some sense, the first step of proving \Cref{thm:santa} is addressing precisely the problem of identifying the correct scales: 
We show that, at a loss of $1+\bigO(\eps)$ in the approximation guarantee, the problem can be reduced to the case of polynomially bounded processing times that are all powers of $1+\eps$. 
To do so, we define suitable (non-trivial) subproblems and assemble them to a global solution with a dynamic program (DP). This approach can be seen as a simpler version of our approach for polynomially bounded processing times, which we focus on in the~following.

Our general approach is to divide the range of possible bag sizes into intervals that contain $\bigO_\eps(1)$ possible approximate bag sizes each. For each such interval, 
it is then possible to guess the set of bags of the respective size in polynomial time. Considering the same range for sizes of \emph{jobs}, rather than bags, we would also be able to guess the assignment of these jobs to these bags.
Observe that a job may, of course, be assigned to a bag whose size lies in a different interval than the job's size. However, we argue that the precise assignment is only relevant when these intervals are neighboring intervals and, hence, can be guessed in polynomial time. 
If this is not the case, i.e., if jobs are assigned to a bag whose size lies in a much larger interval, then such jobs are sufficiently small for us to only consider their total volume.
The resulting parameters, such as number of bags created so far and the assignment of smaller jobs to larger intervals, through which the subproblems corresponding to the intervals interact, are kept track of by a DP. While it is straightforward to keep track of all bags from larger intervals as well as the assignment of jobs from these intervals to the bags, it is not clear how to do this with a polynomial-time DP. 
In particular, when we consider bag sizes of one interval, we still need to remember previously defined bags of much larger sizes with a super-constant number of possibilities for these sizes.
In fact, we show that it is sufficient to keep track of a constant number of parameters that capture all necessary information about larger intervals. Using that the processing times are polynomially bounded, we can bound the size of the DP table by a polynomial in the encoding of the input.

When considering a DP cell corresponding to some interval and guessing bag sizes along with the other parameters implied by the discussion above, we need to evaluate the quality of this guess. To do so, we guess additional parameters (also kept track of by the DP). The main observation is as follows. Suppose that, in addition to the aforementioned parameters, we know the relevant range of the number of machines and the bag sizes from the next-lower interval. For each number of machines in the aforementioned range, we
assign each bag
\begin{enumerate}
    \item[(i)] from a higher interval to one machine each,
    \item[(ii)] from the two currently relevant intervals optimally, and
    \item[(iii)] from the lower intervals fractionally (the total volume of these bags can be approximated). 
\end{enumerate}
The reason we may do so is that the bags assigned in (i) are large enough to
assign enough load to
an entire machine and the bags assigned in (iii) are small enough to be considered~fractional.

We remark that for both problems considered, a PTAS is the best possible approximation algorithm achievable when the number of bags (and machines) is part of the input, unless $\text{P} = \text{NP}$: Since their strongly $\text{NP}$-hard deterministic counterparts~\cite{garey1978strong} are special cases of the stochastic problems, neither makespan minimization nor Santa Claus on stochastic machines admits fully polynomial time approximation schemes (FPTASs) unless $\text{P} = \text{NP}$. If, however, the number of bags (and machines) is not part of the input, i.e., a constant, a FPTAS can be designed by directly guessing the bag sizes approximately, i.e., up to a factor of $1+\epsilon$, in polynomial time and using known FPTASs to compute a job-to-bag assignment based on these bag sizes~\cite{epstein2010maximizing,sahni1976algorithms}.

Stein and Zhong~\cite{SteinZ20} also considered a third objective function, minimizing the difference between the maximum and the minimum machine load. Any polynomial-time approximation algorithm (in the multiplicative sense) is, however, impossible here unless $\text{P} = \text{NP}$. Indeed, already in the deterministic case, it is strongly NP-hard to decide whether the optimal objective-function value is 0 (as can be seen, e.g., by a straightforward reduction from 3-Partition).

\subsection{Further Related Work}

We first review the literature on the aforementioned information-theoretic question in which one compares with the in-hindsight optimum. Since this benchmark is stronger than ours and the upper bounds are obtained through polynomial-time algorithms, the upper bounds carry over to our setting as guarantees of polynomial-time approximation algorithms. Specifically, for makespan, Stein and Zhong~\cite{SteinZ20} showed how to compute for any $\eps>0$ a job-to-bag assignment whose cost is guaranteed to be a factor of at most $5/3+\eps$ away from the cost of the in-hindsight optimum. They also showed an impossibility of $4/3$. When all jobs have infinitesimal size, the best-possible guarantee is $(1+\sqrt{2})/2\approx 1.207$~\cite{EberleHMNSS23,SteinZ20}. For Santa Claus and infinitesimal jobs, the best-possible guarantee is $2\ln2\approx1.386$~\cite{SteinZ20}.

This model has been generalized in two directions. First, Eberle et al.~\cite{EberleHMNSS23} considered arbitrary machine \emph{speeds} (rather than just $0$ or $1$) that are revealed after the bags have been created. They gave a guarantee of $2-1/m$ with respect to the in-hindsight optimum and improved guarantees for special cases. Second, Balkanski et al.~\cite{BalkanskiOSW22} considered the problem with arbitrary speeds in the algorithms-with-predictions framework.

In the majority of the scheduling literature, stochastic uncertainty refers to uncertainty in the processing times of the jobs (see~\cite{Sethuraman16} for a survey and~\cite{EberleGMMZ23,Gupta0NS21,GuptaMZ23} for some recent works). The literature on stochastic uncertainty, even uncertainty in general, in the machines is much more scattered. Stadje considered the unrecoverable breakdown on a single machine caused by stochastic jobs~\cite{Stadje95}. Temporary machine unavailability has also been studied in~\cite{DiedrichJST09}.

\section{Preliminaries}

Formally, we are given a job set~$J = [n] := \{1,\ldots,n\}$, where each job has a \emph{processing time} $p_j$, and a maximum number of machines~$M$. For each $1 \leq m \leq M$, we are given its \emph{probability}~$q_m$, where~$\sum_{m=1}^M q_m = 1$. We want to find a partition of the job set~$J$ into~$M$ sets, called \emph{bags}. We denote the set of bags by~$\bags$. For a bag~$B \in \bags$, let~$p(B) = \sum_{j \in B} p_j$ denote its \emph{size}. We typically say for $j \in B$ that $j$ is \emph{packed} in bag~$B$.

Clearly, if~$M \geq n$, we can pack every job in its own private bag, and the problem becomes trivial. Hence, we assume from now on that $M < n$. 

We denote by $\opt(\bags,m)$ the optimal objective function value for a given set of bags $\bags$ and a scenario with $m$ machines, that is, $\opt(\bags,m)$ denotes the maximum or minimum machine load of an optimal bag-to-machine assignment, or \emph{schedule}, respectively.  
The objective is to find a partition or set of bags~$\bags$ that optimizes $\sum_{m=1}^M q_m \opt(\bags,m)$. We denote a fixed optimal set of bags
by~$\optbags$ and its objective function value by $\opt := \sum_{m=1}^M q_m \opt(\optbags,m)$. 

As discussed above, the problems we consider are generalizations of strongly NP-hard problems.
Thus, unless $\text{P} = \text{NP}$, we cannot expect to find~$\optbags$ in polynomial time. Hence, we are interested in \emph{polynomial-time approximation schemes} (PTASs), i.e., for each~$\eps >0$, a polynomial-time $(1+\eps)$-\emph{approximation algorithm}. Such an algorithm is required to return a partition of the job set~$J$ into~$M$ bags, denoted by~$\bags$, that satisfies 
$\sum_{m=1}^M q_m \opt(\bags,m) \leq (1+\eps) \opt$ for makespan, and $ \sum_{m=1}^M q_m \opt(\bags,m) \geq \frac1{1+\eps} \opt$ for Santa Claus.

\section{Minimizing the maximum machine load}\label{sec:makespan}

In this section we design and analyze our polynomial-time approximation scheme for the setting of makespan minimization: 
For a given number of machines~$m$ and a set~$\bags$ of bags, we want to find an assignment of bags to machines that minimizes the maximum total size of bags assigned to any machine. 

\subsection{Algorithm}\label{subsec:Alg}

Let $\eps>0$; we will give a polynomial-time algorithm that achieves an approximation ratio of $1+\bigO(\eps)$. This algorithm finds a good estimate of the optimal bag sizes in~$\optbags$. To this end, we show later that the maximum size of a bag in~$\optbags$ is at most~$4C$, where 
\[
    C := \sum_{m=1}^M q_m \max \Big\{ \max_{j \in J} p_j, \, \frac1m \sum_{j \in J} p_j \Big\}. 
\]
We say a bag~$B$ in $\optbags$ is \emph{regular} if its size is at least~$\eps C$ 
or if there is at least one job of size at least~$\eps^2 C$ packed in~$B$. 
For~$\ell \in \calL := \{ \lfloor \log_{1+\eps} (\eps^2 C) \rfloor, \lfloor \log_{1+\eps} (\eps^2 C) \rfloor+1, \ldots,  \lceil \log_{1+\eps}(4C) \rceil \} $, the algorithm \emph{guesses} the number $M_\ell$ of optimal bags with $p(B) \in \big[(1+\eps)^{\ell},(1+\eps)^{\ell+1} \big)$. 

Further, it \emph{enumerates} all possible numbers $M_{\sand}$ of bags of size at most 
$(1+\eps)\eps C$, called \emph{sand bags}. These sand bags do not directly correspond to optimal bags, but instead can pack all jobs not packed in regular bags in \opt. 

Clearly, a guess $(M_\ell)_{\ell \in \mathcal L}$ combined with $M_{\sand}$ sand bags does not necessarily guarantee that it is \emph{feasible}, i.e., that $M_{\sand} + \sum_{\ell \in \mathcal L} M_\ell \leq M$ and that there is a partition of $J$ into bags such that there are at most $M_\ell$ bags with sizes in $\big[(1+\eps)^{\ell},(1+\eps)^{\ell+1} \big)$ and $M_{\sand}$ bags of size at most~$(1+\eps)\eps C$.
Thus, the algorithm ignores all combinations of $(M_\ell)_{\ell \in \calL}$ and~$M_{\sand}$ with more than~$M$ bags. If the total number of bags is at most~$M$, the algorithm uses the PTAS by Hochbaum and Shmoys~\cite{HochbaumS88} for bin packing with variable bin sizes to check if there is a feasible packing of jobs into the bags as follows: The input is $\eps$ as approximation parameter, an item of size $p_j$ for each job~$j \in J$, $M_{\ell}$ bins of size $(1+\eps)^{\ell+1}$ for any $\ell\in\calL$, and $M_{\sand}$ bins of size~$(1+\eps)\eps C$. If the guess is feasible, the PTAS is guaranteed to return an item-to-bin (here a job-to-bag) assignment that violates the bin sizes by at most a factor~$(1+\eps)$.

If all jobs can be packed by the above PTAS, the algorithm evaluates the current guess by computing a $(1+\eps)$-approximation of $\opt((M_\ell)_{\ell \in \calL\cup\{\sand\}}, m)$, where we overload notation and let $\opt((M_\ell)_{\ell \in \calL\cup\{\sand\}}, m)$ denote the minimum makespan for a set of bags consisting of $M_\ell$ bags of size $(1+\eps)^{\ell+1}$, for any $\ell\in\calL$, and $M_{\sand}$ bags of size~$(1+\eps)\eps C$.
We denote the makespan of this~$(1+\eps)$-approximation by~$z((M_\ell)_{\ell \in \calL\cup\{\sand\}},m)$ and compute it by running the PTAS by Hochbaum and Shmoys~\cite{HochbaumS87} for makespan minimization on identical machines with approximation parameter~$\eps$, $M_\ell$ jobs with processing time $(1+\eps)^{\ell+1}$, $M_{\sand}$ jobs with processing time~$(1+\eps)\eps C$, and~$m$ machines. 

The algorithm returns a feasible 
minimizer of $\sum_{m=1}^M q_m z((M_\ell)_{\ell \in \calL\cup\{\sand\}},m)$.

\subsection{Analysis}

In this section, we analyze the algorithm designed in the previous section. We start by justifying our bound on the maximum bag size before we argue that there exists a guess that is similar to the optimal set $\optbags$ in terms of the bag size and objective-function value. Last, we evaluate the running time of the algorithm and conclude with the proof of \Cref{theo:makespan}. For formal proofs see \Cref{app:makespan}.

We begin by justifying our assumption to only consider bags of size at most $4 C = 4\sum_{m=1}^M q_m \max \big\{ \max_{j \in J} p_j, \, \frac1m \sum_{j \in J} p_j \big\}$: By \cite{SteinZ20}, $4C$ is an upper bound on~$\opt$. As the largest bag size lower bounds $\opt(\bags,m)$ in scenario~$m$, this implies the next lemma.  
\begin{restatable}{lemma}{lemmakespanlargestbagsize}\label{lem:largest-bag-size}
    No optimal solution uses bags of size greater than $4C$.
\end{restatable}

Fix a set of bags~$\optbags$ 
with objective-function value $\opt$. By \Cref{lem:largest-bag-size}, the maximum bag size is at most~$4C$. The algorithm guesses a set of bag sizes similar to the bag sizes in~$\optbags$. 

Based on $\optbags$ we define a ``good'' guess $(\hat M_\ell)_{\ell \in \calL\cup\{\sand\}}$, i.e., a set of possible bag sizes, as follows: Let~$\optbags[R]$ denote the set of regular bags, i.e., the set of bags in $\optbags$ that pack at least one job of size at least~$\eps^2 C$ or have size at least~$\eps C$. 
\begin{itemize}
    \item For $\ell \in \calL$, let $\hat M_\ell$ be the number of regular bags in $\optbags[R]$ with $p(B) \in \big[(1+\eps)^{\ell},(1+\eps)^{\ell+1} \big)$.
    \item Set  $\hat M_{\sand} = \Big\lceil \frac{\sum_{B \in \optbags\setminus\optbags[R]} p(B) }{\eps C} \Big\rceil$; recall that sand bags have size at most  
    $(1+\eps)\eps C$. 
\end{itemize}

Since the sizes of regular bags are only rounded up, the bags in $(\hat M_\ell)_{\ell \in \calL}$ can pack the same subset of jobs as $\optbags[R]$. 
Since the volume of any sand bag has been increased by a $(1+\eps)$-factor as opposed to $\eps C$ and the size of any job not packed in a regular bag is at most $\eps^2 C$, we show that the bag-size vector $\hat \guess : = (\hat M_\ell)_{\ell \in \calL\cup\{\sand\}}$ is a possible and feasible guess.

\begin{restatable}{lemma}{lemmakespanfeasibleguesspackswell}
\label{lem:feasible-guess_packed-well}
        The above defined vector $\hat \guess$ is a feasible guess of the algorithm. The jobs can be packed into a set of bags consisting of $\hat M_\ell$ bags with size $(1+\eps)^{\ell + 1}$ for~$\ell \in \calL$ and $\hat M_{\sand}$ bags with size~$(1+\eps)\eps C$.
\end{restatable}

Combining this lemma with Theorem~$2$ of \cite{HochbaumS88}, we get the next corollary. 

\begin{corollary}\label{cor:HS_bin-packing-PTAS}
    The PTAS by~\cite{HochbaumS88} returns a packing of all jobs into a set of bags with $M_\ell$ bags of size $(1+\eps)^{\ell + 2}$ for~$\ell \in \calL$ and $M_{\sand}$ bags of size~$(1+\eps)^2\eps C$. 
\end{corollary}

To prove the next lemma, we first assign the regular bags in the same way as in the optimal solution and assign the sand bags one by one to the currently least loaded machine. For bounding the makespan in a given scenario~$m$, we distinguish whether a regular bag determines the makespan (which increases the makespan by at most a factor $(1+\eps)$ compared to $\opt(\optbags,m)$) or whether a sand bag determines the makespan. In the latter case, we use a volume bound to upper bound this sand bag's starting time (losing at most a factor $(1+\eps)$) 
and amortize its maximum size, i.e., $(1+\eps)\eps C$, over all scenarios, using that $\sum_{m=1}^M q_m =1$. 

\begin{restatable}{lemma}{lemmakespangoodguess}
\label{lem:good-guess}
        $\hat \guess = (\hat M_\ell)_{\ell \in \calL\cup\{\sand\}}$ satisfies
    \(
        \sum_{m = 1}^M q_m {\opt}(\hat \guess, m) \leq (1+5\eps){\opt}. 
    \) 
\end{restatable}

\subparagraph*{Proof of \Cref{theo:makespan}.}
Using that we return the cheapest guess and that the number of distinct rounded sizes of regular bags is bounded by $\bigO\big( \frac1{\eps^2} \big)$, we can show the following two lemmas. Combined, they complete the proof of \Cref{theo:makespan}. 

\begin{restatable}{lemma}{lemmakespanapproxratio}
\label{lem:approx-ratio}
    The set of bags returned by the algorithm guarantees an objective function value of at most $(1+\bigO(\eps)) \opt$. 
\end{restatable}

\begin{restatable}{lemma}{lemmakespanrunningtime}
\label{lem:running-time}
        For $\eps \in \big(0,\frac12\big)$, the algorithm runs in time $\bigO \Big( \big(\frac n {\eps}\big)^{\bigO(1/\eps^2)} \Big)$. 
\end{restatable}

\section{Maximizing the minimum machine load}\label{sec:santaclaus}

In this section, we present our polynomial-time $(1+\epsilon)$-approximation
algorithm for the setting in which we want to maximize the minimum
machine load. The formal proofs for the results presented in this section can be found in~\Cref{app:santa}.

\subsection*{Polynomially bounded processing times}

First, we show that we can reduce our problem to the case of polynomially bounded job processing times that are all essentially powers
of $1+\epsilon$, while losing at most a factor of $1+\bigO(\epsilon)$ in our approximation guarantee.
The main concepts of the reduction can be summarized by the following three ideas.

The first idea is to disregard scenarios whose contribution to the expected objective function value is very small.
W.l.o.g., assume that $p_j \in \mathbb{N}$ and let $d$ be an integer such that $\opt(\optbags,m)$ falls in the interval $\big[1,\big(\frac{n}{\epsilon}\big)^d\big]$ for every scenario $m$.
Then, for some offset $a \in \big\{0,1,\dots,\frac{1}{\epsilon} + 3\big\}$ we ``split'' the interval $\big[1,\big(\frac{n}{\epsilon}\big)^d\big]$ into a polynomial number of pair-wise disjoint intervals $\Tilde{I}_i = \Big[\left( \frac{n}{\epsilon} \right)^{3k + \frac{k-1}{\epsilon} + a}, \left( \frac{n}{\epsilon} \right)^{3k + \frac{k}{\epsilon} + a}\Big)$.
Observe that any two consecutive intervals have a multiplicative gap of $\left( \frac{n}{\epsilon} \right)^3$.
Using probabilistic arguments, we show that there is an offset $a$ such that the scenarios with $\opt(\optbags,m)$ in the gaps contribute very little to the expected objective function value.
Hence, such scenarios can be neglected by losing a factor of at most $1+\bigO(\epsilon)$ in the approximation ratio.
As there is only a polynomial number of possible offsets $a$, we may assume that we correctly choose such $a$ by enumeration. 

The second idea is to observe that the gaps enable us to actually \emph{ignore} a carefully chosen subset of jobs.
Let $\Tilde{I}^+_i = \Big[\left( \frac{n}{\epsilon} \right)^{3k + \frac{k-1}{\epsilon} + a - 3},\left( \frac{n}{\epsilon} \right)^{3k + \frac{k-1}{\epsilon} + a}\Big) \cup \Tilde{I}_k$ denote the extended interval obtained by the union of the interval $\Tilde{I}_k$ and the smaller of its adjacent gaps, and let $m$ be a scenario such that $\opt(\optbags,m) \in \tilde{I}_k$.
We show that, by losing a factor of at most $1+\bigO(\epsilon)$ in the approximation ratio, we may assume that a machine with minimum load in the schedule that achieves $\opt(\optbags,m)$ is assigned only bags with jobs whose processing time is in $\Tilde{I}^+_k$.
Then, based on this assumption, we show that we may assume that there are no jobs whose processing times fall in the gaps by losing at most another factor of $1+\bigO(\epsilon)$ in the approximation ratio. 
Observe that we are now facing an instance where neither $\opt(\optbags,m)$ nor $p_j$ belong to the just created gaps in the interval $\big[1,\big(\frac{n}{\epsilon}\big)^d\big]$.

The third idea is then to solve the problem restricted to the intervals $\tilde I_k$ individually by using the fact that within each interval $\tilde I_k$ the processing times are polynomially bounded and combine the obtained solutions into a single one with a dynamic program.
We show that rounding up the processing times and solving each subproblem that arises in the dynamic program costs a factor of at most $1+\bigO(\epsilon)$ in the approximation ratio.
Formalizing this proof sketch proves~\Cref{lem:rounding}.

\begin{restatable}{lemma}{lemrounding}
\label{lem:rounding}
    By losing a factor of at most $1+\bigO(\epsilon)$ in the approximation ratio, we can assume for each job~$j \in J$ that $p_{j}=\lceil (1+\epsilon)^{\ell_j}\rceil$ for some $\ell_j\in \mathbb{N}_0$ and $p_{j}\in[1,n^{c(\epsilon)}]$ where $c(\epsilon)$ only depends on $\eps$.
\end{restatable}

\subparagraph*{Algorithmic overview.}

Based on \Cref{lem:rounding}, we assume that each job $j\in J$ satisfies~$p_{j}\in[1,n^{c(\epsilon)}]$. The high-level idea of our algorithm is to partition $[1,n^{c(\epsilon)}]$ into intervals of the form $I_{k}:=\big[\big(\frac{1}{\epsilon}\big)^{3k},\big(\frac{1}{\epsilon}\big)^{3k+3}\big)$ for $k \in \mathbb N$ and only consider bags, jobs, and scenarios relevant for a \emph{single} interval. More precisely, we use these intervals to partition the processing times $\left\{ p_{j}\right\} _{j\in J}$, the bag sizes in $\optbags$ and in our solution as well as the values $\left\{ \opt(\optbags,m)\right\} _{m}$. Let $\kmax$ such that $\sum_{j\in J}p_{j}\in I_{\kmax}$; we ignore intervals $I_{k}$ with $k>\kmax$. For $k\in [\kmax]$, let $J_{k}:=\{j\in J:p_{j}\in I_{k}\}$, let $L_k := \big\{\ell \in \mathbb N: \big\lceil (1+\eps)^\ell \big\rceil \in I_k \big\}$, and let $\optbags[k] := \{B \in \optbags: p(B) \in I_{k}\}$.

Our algorithm recursively considers the intervals in the order $I_{\kmax},I_{\kmax-1},...,I_{1}$
and, step by step, defines bags that correspond to $\optbags[\kmax],\optbags[\kmax-1],...,\optbags[1]$. When considering interval~$I_k$, the algorithm enumerates all possible bag sizes of bags in $\optbags[k]$ and all possible assignments of a subset of the jobs in $J_k \cup J_{k-1}$ to those bags; the remaining jobs in $J_k$ are implicitly assigned to bags in $\bigcup_{k' = k+1}^{\kmax} \optbags[k']$. 
Here, we use the fact that by definition of our intervals only a constant number of jobs in $J_k \cup J_{k-1}$ can be assigned to any bag in $\optbags[k]$ while jobs in $J_k$ are tiny compared to bags in $\bigcup_{k' = k+2}^{\kmax} \optbags[k']$ (see \Cref{fig:jobvsbags} for visualization) and, hence, the assignment of jobs $J_k$ to bags $\bigcup_{k' = k+2}^{\kmax} \optbags[k']$ cannot be guessed in polynomial time.
The remaining jobs in $J_{k-1}$ will be assigned when interval $I_{k-1}$ is considered. 
We embed this recursion into a polynomial-time dynamic program (DP). 
Since our DP is quite technical, we first describe the algorithmic steps that correspond to the root subproblem of the recursion, i.e., $I_{\kmax}$, before we define the DP cells and argue about their solution.
\begin{figure}
    \centering
    \includegraphics[width = 0.9\linewidth]{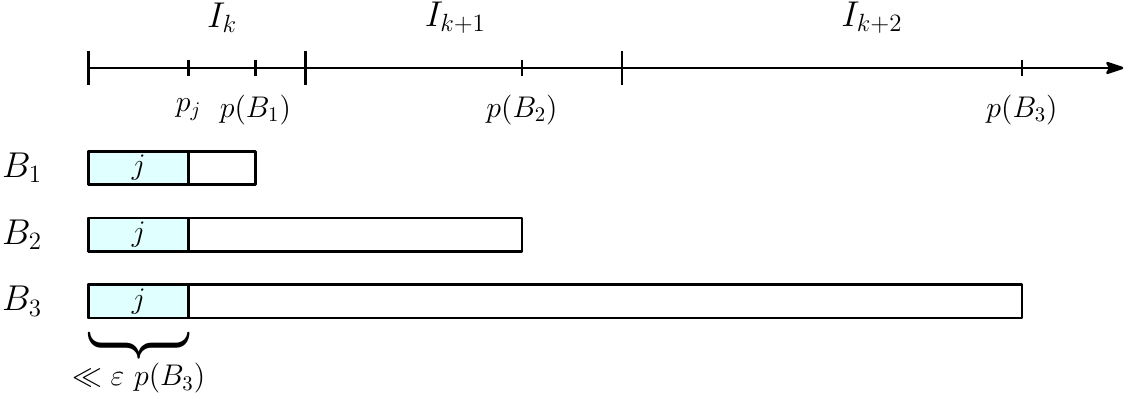}
    \caption{Visualization of relation between jobs in $J_k$ and bags in $\optbags[k]$, $\optbags[k+1]$ and $\bigcup_{k' = k+2}^{\kmax} \optbags[k']$.} 
    \label{fig:jobvsbags}
\end{figure}
Defining the DP cells and solving their corresponding subproblem involves enumerating all possible values of several quantities and storing an approximation of the objective-function value of the (approximately) best combination in the DP cell.
When arguing about the correctness of our DP, we show that there is a chain of DP cells that represent some (fixed) optimal solution. 
Hence, we use~$X^*$ for some parameter~$X$ when referring to the correct value, i.e., 
the value of this parameter in this optimal solution. 
We refer to this process as \emph{guessing}~$X^*$. 
In general, we use $\hat X$ to refer to an arbitrary guess. 

\subsection{\label{subsec:Root-subproblem} Guessing initial quantities}
In the following, we describe how to guess and evaluate initial parameters corresponding to a (partial) solution of our root subproblem. Intuitively, we construct a partial assignment of jobs to bags $\optbags[K]$ and the parameters representing this partial assignment define the DP cell corresponding to the remaining problem.
\subsubsection{Algorithm}
The algorithm to compute a partial solution to a root subproblem can be essentially split into two phases: (1) \emph{guessing} key quantities and (2) \emph{evaluating} these guesses.
\subparagraph*{Guessing phase.}
We start by guessing $|\optbags[\kmax]|$ and $|\optbags[\kmax-1]|$, the number of bags in $\optbags[\kmax]$ and in $\optbags[\kmax-1]$, before we guess $(1+\epsilon)$-approximations for the bags sizes in $\optbags[\kmax]\cup\optbags[\kmax-1]$. 
Formally, for each bag $B\in\optbags[\kmax]\cup\optbags[\kmax-1]$ we guess a value $\ell(B)\in\N$ such that $p(B)\in[(1+\epsilon)^{\ell(B)},(1+\epsilon)^{\ell(B)+1})$; we say that such a value $\ell(B)$ is the \emph{size-estimate}
for $B$.
Next, we guess an assignment of all jobs in $J_{\kmax}$ and a subset of the jobs in $J_{\kmax-1}$ to the bags $\optbags[\kmax]$ and an assignment of the remaining jobs in $J_{\kmax-1}$ and of a subset of the jobs in $J_{\kmax-2}$ to the bags $\optbags[\kmax-1]$. 
Finally, we guess $m_{\max}^{(\kmax)}$ which we define to be the largest value $m\in M$ for which $\opt(\optbags,m)\in I_{\kmax}$. 

\subparagraph*{Evaluation phase.}
In contrast to the previous section, maximizing the minimum machine load asks for ``covering'' a machine or, in our case, a bag. 
To this end, we potentially have to assign jobs from $\bigcup_{k'=1}^{\kmax-2} J_{k'}$ to the bags in $\optbags[\kmax]$. 
Formally, for $B\in\optbags[\kmax]$ let $p^{+}(B)$ be the total size of the jobs from $J_{\kmax}$ and $J_{\kmax-1}$ already assigned to~$B$. 
We define $S:=\sum_{B\in\optbags[\kmax]}\max\{\left\lceil (1+\epsilon)^{\ell(B)}\right\rceil - p^{+}(B),0\}$.
Our DP also stores this value in order to guarantee that, in the remainder, jobs from $\bigcup_{k'=1}^{\kmax-2} J_{k'}$ with total size~$S$ are assigned to bags in $\optbags[\kmax]$. 
Let $J_{\kmax-1}(\optbags[\kmax-1])$ and $J_{\kmax-2}(\optbags[\kmax-1])$ be the subsets of $J_{\kmax-1}$ and $J_{\kmax-2}$ already assigned to bags in $\optbags[\kmax-1]$. 
Then, $\bar S:= \sum_{B\in \optbags[\kmax-1]} \left\lceil (1+\epsilon)^{\ell(B)}\right\rceil - p\big(J_{\kmax-1}(\optbags[\kmax-1]) \cup J_{\kmax-2}(\optbags[\kmax-1])\big)$ is the total volume of bags in $\optbags[\kmax-1]$ that needs to be covered with jobs from $\bigcup_{k'=1}^{\kmax-3} J_{k'}$. 

For evaluating our current guess, we fix some $m \leq m^{(\kmax)}_{\max}$ and create a set $J_T$ of dummy jobs, each with processing time~$1$ and total size $T:=\sum_{k=1}^{\kmax -2} \sum_{ j \in J_k} p_j - S - \bar{S}$.
Now, we guess the assignment of the bags $\optbags[\kmax] \cup \optbags[\kmax-1]$ to the machines.
Based on the load guaranteed by these bags, we now greedily distribute these dummy jobs as follows. 
Assume w.l.o.g. that the machines are sorted non-decreasingly by their loads and consider the prefix of the machines which all have the smallest load. 
We assign to each of these machines the same number of dummy jobs such that their new load is equal to the load of the machines with the second smallest load. We repeat this procedure until all dummy jobs in $J_T$ are assigned.
At the end, among all possibilities to assign the bags $\optbags[\kmax] \cup \optbags[\kmax-1]$, we choose the one which maximizes the minimum machine load after the distribution of $J_T$. 
We define $\ALG(m)$ as the load of the least loaded machine for this fixed candidate solution.

Among all guesses with the same set of bag sizes for bags in $\optbags[\kmax]\cup \optbags[\kmax-1]$, the same value~$m_{\max}^{(\kmax)}$ and the same values~$S$ and~$\bar S$, we keep the guess which maximizes our proxy for the (partial) objective function, $\sum_{m=1}^{m_{\max}^{(\kmax)}} q_m \cdot \ALG(m)$. 

\subsubsection{Analysis}
Observing that $|\optbags[\kmax]| \leq M \leq n$ and $|\optbags[\kmax-1]|\leq M \leq n$ implies that we can enumerate all possible combinations in time $O(n^2)$. Since the relative length, i.e., the ratio of the left interval border to the right interval border, of $I_{\kmax} \cup I_{\kmax-1}$ is bounded by $\big(\frac1\eps\big)^{6}$, there are at most $\bigO_{\epsilon}(1)$ possibilities for each size-estimate $\ell(B)$. By guessing the number of bags with a given size estimate, we can guess the size-estimates of all bags in $\optbags[\kmax]\cup\optbags[\kmax-1]$ in time $n^{\bigO_\eps(1)}$.
Further, each bag in $\optbags[\kmax]$ can be assigned at most $O_{\epsilon}(1)$ many jobs from $J_{\kmax} \cup J_{\kmax-1}$ and, similarly, each bag in $\optbags[\kmax-1]$ can be assigned at most $O_{\epsilon}(1)$ many jobs from $J_{\kmax-1} \cup J_{\kmax-2}$. 
Hence, there is only a constant number of possible assignments for each bag, up to permutations of jobs with the same size.
We formalize these observations in the next lemma.\footnote{For the initial guesses, one could give tighter bounds by observing that $|\optbags[\kmax]|+|\optbags[\kmax]-1|=\bigO_\eps(1)$. However, we give polynomial bounds which are sufficient and of the same kind as the bounds we will use later in the DP.}

\begin{restatable}{lemma}{lemscrootsizeassign}
\label{lem:sc_root_size}
    In time $n^{\bigO_{\epsilon}(1)}$, we can guess the size-estimate $\ell(B)$ for each bag $B\in\optbags[\kmax]\cup\optbags[\kmax-1]$ as well as the assignment of the jobs in $J_{\kmax}$ to the bags $\optbags[\kmax]$, of the jobs in $J_{\kmax-1}$ to the bags $\optbags[\kmax]\cup\optbags[\kmax-1]$ and of a subset of jobs in $J_{\kmax-2}$ to the bags in $\optbags[\kmax-1]$, up to a permutation of bags.
\end{restatable}

First, observe that for each bag $B \in \optbags[\kmax]$, the value $\max\{\lceil (1+\eps)^{\ell(B)}\rceil - p^{+}(B), 0\} \in \mathbb N_0$ since $p_j \in \mathbb{N}$ for each $j \in J$ by \Cref{lem:rounding}.
Hence, $S$ accurately captures the total volume missing to ensure that each $B \in \optbags[\kmax]$ packs jobs with a total size of at least $(1 + \eps)^{\ell(B)} \geq \frac{p(B)}{1+\eps}$. 
Using that each job in $\bigcup_{k=1}^{\kmax-2} J_{k}$ is very small compared to a bag in~$\optbags[\kmax]$, we can argue that knowing $S$ is actually sufficient to cover $B \in \optbags[\kmax]$ with jobs of total size of at least $\frac{p(B)}{(1+\eps)^2}$. 

Similarly, for a bag in $\optbags[\kmax-1]$, each job in $\bigcup_{k=1}^{K-3} J_k$ is very small compared to its size. Hence, we can again argue that knowing $\bar S \in \mathbb N_0$ is sufficient to pack jobs of total size at least $\frac{p(B)}{(1+\eps)^2}$ into bag $B \in \optbags[\kmax-1]$. 

Observing that no bag in $\bigcup_{k=1}^{\kmax-2} \optbags[k]$ can pack a job from $J_{\kmax}\cup J_{\kmax-1}$ by definition of their sizes, we conclude that $T$ indeed represents the total volume of bags in $\bigcup_{k=1}^{\kmax-2} \optbags[k]$.
In fact, we can show that for scenarios with $m\le m_{\max}^{(\kmax)}$ machines 
\emph{any} assignment of the remaining jobs in $\bigcup_{k'=1}^{\kmax-2} J_{k'}$ of total volume at most $\frac T{1+\eps}$ to at most $M - |\optbags[\kmax]| - |\optbags[\kmax-1]|$ bags of size at most $\eps \big(\frac1\eps\big)^{3\kmax}$ yields the same objective function value (up to a factor of $1+\bigO(\eps)$).
These observations are formalized in the next lemma where some jobs are set aside in bags $B_S$ and $B_{\bar S}$, corresponding to the values~$S$ and~$\bar S$.

Recall that we use $\hat X$ to denote a possible guess for parameter $X$ considered by our algorithm. 
\begin{restatable}{lemma}{lemscrootmain}
\label{lem:sc_root_main}
    Let the guessed quantities be as defined above.
    Let $\bags'\cup\{B_{S}, B_{\bar{S}}\}$ be a partition of the jobs $\bigcup_{k'=1}^{\kmax-2} J_{k'}$
    into $M-|\guessbags[\kmax]|-|\guessbags[\kmax-1]|+2$
    bags such that  
    \begin{itemize}
        \item for each bag $B\in\bags'$, $p(B)\le\epsilon\cdot(\frac{1}{\epsilon})^{3\kmax}$,
        \item $p(B_{S})\ge S$,
        \item $p(B_{\bar{S}})\ge \bar{S}$,
        and
        \item $p(\bags'):=\sum_{B\in\bags'}p(B)\ge(1+\epsilon)^{-1}T$.
    \end{itemize}
    Suppose that $B\in\guessbags[\kmax]\cup\guessbags[\kmax-1]$ has size in \mbox{$[(1+\epsilon)^{\ell(B)}, (1+\epsilon)^{\ell(B)+1})$}. 
    We can compute the vector $\left( \ALG(m)\right)_{m=1}^{m_{\max}^{(\kmax)}}$ in polynomial time and ${\opt}(\bags'\cup\guessbags[\kmax]\cup\guessbags[\kmax-1],m)\in [(1+\epsilon)^{-5}\ALG(m),(1+\epsilon)\ALG(m))$ for each $m\le m_{\max}^{(\kmax)}$. 
\end{restatable}

Note that the lemma does not state anything about the relationship of $\opt(\optbags,m)$ and ${\opt}(\bags'\cup\guessbags[\kmax]\cup\guessbags[\kmax-1],m)$; it relates our proxy function $\ALG(m)$ and ${\opt}(\bags'\cup\guessbags[\kmax]\cup\guessbags[\kmax-1],m)$, the best possible assignment for $\bags'\cup\guessbags[\kmax]\cup\guessbags[\kmax-1]$.

In the remaining problem it suffices to focus on scenarios in which we have $m > m_{\max}^{(\kmax)}$ machines and which hence satisfy $\opt(\optbags,m)<(\frac{1}{\epsilon})^{3\kmax}$. 
Note that for each bag $B\in\optbags[\kmax]$ we have that $p(B)\ge(\frac{1}{\epsilon})^{3\kmax}$. 
Therefore, if we are given $m > m_{\max}^{(\kmax)}$ machines, it is optimal to assign each
bag $B\in\optbags[\kmax]$ to a separate machine without any further bags assigned to that machine. 
Hence, if $m > m_{\max}^{(\kmax)}$, then the bags in $\optbags[1],...,\optbags[\kmax-1]$ need to ensure only
that the remaining $m-|\optbags[\kmax]|$ machines get enough load. 
This insight and the above lemma allow us to decouple our decisions for scenarios with $m \le m_{\max}^{(\kmax)}$ machines from scenarios with $m > m_{\max}^{(\kmax)}$ machines. 
This is the key idea for our DP.

\subsection{Dynamic program}
After our initial guesses above, it remains to
\begin{itemize}
\item pack the jobs in $\bigcup_{k'=1}^{\kmax-1} J_{k'}$ into the bags in $\optbags[\kmax-1]$,
\item compute the bag sizes in $\optbags[1],...,\optbags[\kmax-2]$, and 
\item select jobs from $ \bigcup_{k=1}^{\kmax-2} J_k$ with total size at least $S$ for filling $\optbags[\kmax]$.
\end{itemize}

For each $m\ge m_{\min}^{(\kmax-1)}:=m_{\max}^{(\kmax)}+1$, our goal is to obtain a value close enough to $\opt(\optbags,m)$ so that, overall, we achieve a value of $(1-\bigO(\epsilon))\opt$.
 
To this end, we design a dynamic program (DP) that solves the remaining problem from above. Each DP cell corresponds to some subproblem. We show that for each possible guess of the initial quantities in \Cref{subsec:Root-subproblem} there is a DP cell corresponding to the remaining subproblem.
In order to solve each subproblem, we guess similar quantities as in the previous section and
reduce the resulting remaining problem to the subproblem of another DP cell.

\subsubsection{Algorithm}
Following the same idea as for the root subproblem, our dynamic program proceeds as follows: for each DP cell we first \emph{guess} key quantities defining a partial solution as well as the transition to the next DP cell and then we \emph{evaluate} this guessed partial solution.

\subparagraph*{DP cell and its subproblem.} Formally, each DP cell~$\mathcal{C}$ is specified by
\begin{itemize}
\item $k\in\N$ with $k<\kmax$ specifying that we still need to define $\bags_{1},...,\bags_{k}$,
\item $\noBags \in\N$ counting the previously defined (large) bags $\bags_{k+1},...,\bags_{\kmax}$,
\item $M_{k}\in\N$ representing our decision $|\bags_{k}|=M_{k}$, 
\item $m_{\min}^{(k)}\in\N$ indicating the minimal number of machines we consider,
\item $s_{\ell}\in\N$ for $\ell\in L_k$ counting $B \in \bags_k$ with $\ell(B)=\ell$,
\item $a_{\ell} \in\N $ for $\ell\in L_k $ counting the jobs~$j$ with $p_{j}=\lceil(1+\epsilon)^{\ell}\rceil$ that are assigned to bags in $\bags_{k+1}$,
\item $S\in\N$, 
defining the total size of jobs in $\bigcup_{k'=1}^k J_{k'}$ that 
must not be assigned to bags $\bigcup_{k'=1}^k \bags_{k'}$ and that are neither assigned to bags in $\bags_{k+1}$ via the values $a_\ell$; instead they will be assigned to $\bigcup_{k'=k+1}^{\kmax} \bags_{k'}$. (Note that we will make sure that even though jobs from $J_k$ might contribute to $S$, such jobs will not be used to cover bags in $\bags_{k+1}$.)
\end{itemize}

The goal of the subproblem of $\mathcal C$ is to pack a subset of the jobs $\bigcup_{k'=1}^k J_{k'}$ 
into the bags $\bigcup_{k'=1}^{k} \bags_{k'}$ and define 
a size-estimate $\ell(B)$ for $B\in\bigcup_{k'=1}^k \bags_{k}$ 
such that
\begin{itemize}
\item $|\bags_{k}|=M_{k}$,
\item $p(B) \in I_{k'}$ for each $k'\in [k]$ and each $B\in \bags_{k'}$,
\item $p(B)\in [(1+\epsilon)^{\ell(B)},(1+\epsilon)^{\ell(B)+1})$ for each $B\in \bigcup_{k'=1}^k \bags_{k}$,
\item there are $s_{\ell}$ bags $B\in\bags_{k}$ with $\ell(B)=\ell$ for each $\ell\in L_k$, 
\item each job $j\in J_{k'}$ for $k'\in[k]$ is either
 assigned to some bag in $\bags_{k'}\cup... \cup \bags_{k}$ or not at all,
\item there
are $a_{\ell}$ jobs $j$ with $p_{j}=\lceil(1+\epsilon)^{\ell}\rceil$ that are not assigned to any bag for each $\ell\in L_k$,
\item the jobs in $\bigcup_{k'=1}^k J_{k'}$ not packed in any
bag have total size at least $S$.
\end{itemize}

For each DP cell $\mathcal C$, we compute a solution and a corresponding objective function value which we denote by $\mathrm{profit}(\mathcal{C})$. 
This objective function corresponds to the expected profit from scenarios in $\{m_{\min}^{(k)},...,M\}$ that we achieve with the solution stored in the DP cell and $\noBags$ ``large'' bags, i.e., bags~$B$ with $p(B) \in \bigcup_{k'=k+1}^{\kmax} I_{k'}$.

\subparagraph*{Guessing phase.}
By definition of the DP cell, $|\optbags[k]|=M_{k}$. 
For each $\ell \in L_k$, there are $s_{\ell}$ many bags $B \in \optbags[k]$ with $\ell(B) = \ell$ (and hence with $p(B) \in [(1+\epsilon)^\ell,(1+\epsilon)^{\ell+1})$). 
We start by guessing the assignment of the jobs $J_{k-1}\cup J_{k}$ to the bags in $\optbags[k]$. 
We only consider guesses satisfying the values $a_\ell$ and $S$ of our current DP cell, i.e., for every possible processing time in $I_k$ enough jobs are left to be assigned to $\optbags[k+1]$ and enough total volume of jobs in $\bigcup_{k'=1}^k J_{k'}$ is left to be assigned to bags in $\bigcup_{k'=k+1}^K \optbags[k']$. 
Finally, we guess $m_{\max}^{(k)}$ which is the largest value $m$ for which $\opt(\optbags,m) \in I_k$.

\subparagraph*{Evaluation phase.}
In order to calculate the proxy objective function value $\profit(\mathcal{C})$, we need to combine $\calC$ with a cell $\hat\calC$ corresponding to a DP cell for the remaining problem. 
To this end, let us define the parameters of this cell $\hat\calC$. 
Clearly, we only need to define $\bags_1, \ldots, \bags_{k-1}$. 
Hence, the first parameter of $\hat\calC$ is $k-1$. 
Further, the total number of previously defined bags is given by $\hat{M}_{k,\ldots,K} = M_k + \noBags$. 
As we do not ignore scenarios, we choose $m_{\min}^{(k-1)}:= m_{\max}^{(k)}+1$. 
Since we have already guessed the assignment of jobs in $J_{k-1}$ to bags in $\optbags[k]$, we can simply calculate the values $\hat{a}_\ell$ for $\ell \in L_{k-1}$ that indicate the number of jobs $j$ with $p_j=\lceil(1+\epsilon)^{\ell}\rceil$ to be assigned to bags in $\optbags[k]$. 

It remains to calculate the value $\bar S \in \N$, the total size of jobs in $\bigcup_{k'=1}^{k-1} J_{k'}$ assigned as very small jobs to bags in $\bigcup_{k'=k}^{\kmax} \optbags[k']$. 
To this end, we calculate the total size of jobs from $\bigcup_{k'=1}^{k-2} J_{k'}$ that need to be packed in $\optbags[k]$.
For each $B\in\optbags[k]$, let $p^{+}(B)$ be the total size of the jobs from $J_{k-1} \cup J_{k}$ that~$B$ already packs. 
We define $S_{k}:=\sum_{B\in \optbags[k]}\max\left\{\left\lceil (1+\epsilon)^{\ell(B)}\right\rceil -p^{+}(B),0\right\}$.
Denote by $J_k(\optbags[k]\cup\optbags[k+1])$ the set of jobs from $J_k$ assigned to bags $\optbags[k]$ and $\optbags[k+1]$.  
Then, $\bar S$ is defined as
$\bar S:= S - \sum_{j\in J_k \setminus J_k(\optbags[k]\cup\optbags[k+1])}p_j + S_k,$
where $S$ is defined by the current DP cell~$\calC$. 
(Note that $\bar S$ does not contain jobs from $J_{k-1}$ to be assigned to $\optbags[k]$ as they are accounted for by $\hat{a}_\ell$.)

Hence, the remaining problem corresponds to some DP cell~$\hat\calC$ satisfying
\begin{equation}
    \hat{\mathcal{C}} = \Big(k-1,\noBags+ M_k ,\hat M_{k-1}, m_{\max}^{(k)}+1, \hat S \geq \bar S,(\hat{s}_\ell)_{\ell \in L_{k-1}}, (\hat{a}_\ell)_{\ell \in L_{k-1}}\Big), \label{eq:sc_dp_next}
\end{equation}
where $\hat{s}_\ell$ is a possible number of bags with size-estimate $\ell \in L_{k-1}$, i.e., with size $(1+\epsilon)^\ell$, the number of bags $|\optbags[k-1]|$ is given by $\hat{M}_{k-1} = \sum_{\ell \in L_{k-1}}\hat{s}_\ell$, and we require that $\hat S$ is at least $\bar S$.  

Given $\profit(\hat\calC)$ for some $\hat\calC$, we can now calculate $\profit(\calC)$ as follows: 
For each value $m \in \{m_{\min}^{(k)},...,m_{\max}^{(k)}\}$, we compute an estimate $\ALG(m)$ of the objective value of an optimal bag-to-machines assignment of $\bigcup_{k'=1}^k \hat{\bags}_{k'}$ and $\noBags$ large bags to $m$ machines. 
To this end, we use a variant of \Cref{lem:sc_root_main}, which is explained in detail in the appendix. 
Then, the profit of the candidate combination of $\calC$ with $\hat\calC$ is given by  
$\sum_{m=m_{\min}^{(k)}}^{m_{\max}^{(k)}}q_m \ALG(m) + \mathrm{profit}(\hat\calC)$.
Among all these candidate combinations, we choose the one with the largest profit and set $\profit(\calC) = \sum_{m=m_{\min}^{(k)}}^{m_{\max}^{(k)}}q_m \ALG(m) + \mathrm{profit}(\hat\calC)$.

\subsubsection{Analysis}

Observe that there are at most $\bigO_{\epsilon}(1)$ many distinct processing times of jobs $J_{k-1}  \cup J_{k}$ and each bag $B \in \optbags[k]$ contains at most $\bigO_{\epsilon}(1)$ many jobs from each of these processing times because of the definition of $\optbags[k]$, $J_{k-1}$, and~$J_k$. 

We guess all possible assignments of jobs to a single bag of size at most $\big(\frac1\eps\big)^{3k+3}$; typically such an assignment is called a \emph{configuration}. 
There are at most $\bigO_{\epsilon}(1)$ such configurations. 
Then, for each configuration and each $\ell$ with $\big\lceil(1+\epsilon)^{\ell}\big\rceil \in I_k$, we guess how often the configuration is assigned to a bag $B$ with $\ell(B) = \ell$. 
Following this sketch, the next lemma proves that we can in fact guess the job-to-bag assignment in polynomial time.

\begin{restatable}{lemma}{lemscassignDP}
    \label{lem:lmscassginDP}
    In time $n^{\bigO_{\epsilon}(1)}$ we can guess the assignment of jobs from $J_{k-1}$ and $J_{k}$ to the bags $\optbags[k]$ up to permuting jobs and bags. 
\end{restatable}

During the evaluation phase, we try all possible combinations of the current DP cell~$\calC$ with $\hat\calC$ satisfying \eqref{eq:sc_dp_next}, i.e., DP cells corresponding to the remaining subproblems matching the parameters of $\calC$. 
We now give a proof sketch of why our guesses combined $\calC$ indeed give a feasible solution to the subproblem for $k$. 
Let $\hat{\bags}_1,\dots,\hat{\bags}_{k-1}$ be the bags given by the solution to $\hat\calC$ and $\guessbags[k]$ be the bags corresponding to our guess. (We do not change $\bigcup_{k'=1}^{k-1} \hat{\bags}_{k'}$.)

We need to assign jobs from $\bigcup_{k'=1}^{k-1} J_{k'}$ to  $\guessbags[k]$ satisfying
\begin{enumerate}[(1)]
    \item for every bag $B \in \guessbags[k]$, $p(B) \geq (1+\epsilon)^{\ell(B)-1}$ and
    \item the total processing time of 
    \begin{itemize}
    \item all jobs in $\bigcup_{k'=1}^{k-1}J_{k'}$ not assigned to bags in $\bigcup_{k'=1}^k \hat{\bags}_{k'}$ and
    \item all jobs in $J_{k}$ neither assigned to $\guessbags[k]$ nor reserved by the values $a_\ell$ for the bags with size in $I_{k+1}$
    \end{itemize}
    is at least $S$.
\end{enumerate}

Each $B \in \guessbags[k]$ has already jobs from $J_k$ and $J_{k-1}$ of total size $p^{+}(B)$ assigned to it. 
Let $p^-(B) := \max\big\{\left\lceil (1+\epsilon)^{\ell(B)}\right\rceil - p^{+}(B), 0\big\}$ be the missing volume in $B$ to cover $B$ to the desired level of $\left\lceil (1+\epsilon)^{\ell(B)}\right\rceil$. 
If $p^-(B) = 0$, no additional job from $\bigcup_{k'=1}^{k-2} J_{k'}$ needs to be assigned to $B$. 
Otherwise, we greedily add jobs from $\bigcup_{k'=1}^{k-2} J_{k'}$ not packed in $\bigcup_{k'=1}^{k-1} \hat\bags_{k'}$ until assigning the next job would exceed $p^-(B)$. 
Hence, the total size of jobs $\bigcup_{k'=1}^{k-2} J_{k'}$ assigned to $B$ by this routine is at least $p^-(B) - \big(\frac{1}{\epsilon}\big)^{3k-6}$. 
Thus, 
$$
    p(B) \geq p^{+}(B) + p^-(B) - \bigg(\frac{1}{\epsilon}\bigg)^{3k-6} = (1+\epsilon)^{\ell(B)} - \bigg(\frac{1}{\epsilon}\bigg) ^{3k-6} \geq (1+\epsilon)^{-1}(1+\epsilon)^{\ell(B)}
$$
since by definition $(1+\epsilon)^\ell \geq \frac{1}{\epsilon} \cdot \big(\frac{1}{\epsilon}\big)^{3k-6}$ for all $\ell \in L_k$.

By choice of $\hat\calC$, the total size $\hat S$ of jobs in $\bigcup_{k'=1}^{k-1} J_{k'}$ neither assigned to bags in $\bigcup_{k'=1}^{k-1} \hat \bags_{k'}$ nor reserved for $\bags_{k}$ via the values $(\hat a_\ell)_{\ell \in L_{k-1}}$ is at least $\bar S$. 
With the definition of $\bar S$, we get 
\begin{equation*}
    \hat S \geq \bar S = S - \sum_{j\in J_k \setminus J_k(\guessbags[k]\cup\guessbags[k+1])}p_j + S_k = S - \sum_{j\in J_k \setminus J_k(\guessbags[k]\cup\guessbags[k+1])} p_j + \sum_{B\in \guessbags[k]} p^{-}(B) \, . 
\end{equation*}
We remark that the second term indeed corresponds to the contribution of jobs from $J_k$ to filling bags with sizes in $\bigcup_{k' = k+1}^{\kmax} I_{k'}$ since such jobs cannot be packed into $\bigcup_{k'=1}^{k-1} \guessbags[k]$ by definition of the corresponding sizes. 
Observe that the greedy procedure described above assigns jobs from $\bigcup_{k'=1}^{k-2} J_{k'}$ with total volume \emph{at most} $p^{-}(B)$ to $B \in \guessbags[k]$. 
Hence, the combination of our guess $\guessbags[k]$ (and the guessed partial assignment of $J_{k-1}\cup J_k$ to $\guessbags[k]$) and the solution for $\hat\calC$ is indeed a feasible solution for $\calC$.

Similar to the proof of \Cref{lem:sc_root_main}, we show that for any candidate solution (consisting of a guess $\guessbags[k]$, a partial assignment of $J_k \cup J_{k-1}$, and any solution for $\hat \calC$ as defined in \eqref{eq:sc_dp_next}), we can calculate the values $\ALG(m)$ for $m \in \{ m_{\min}^{(k)}, \ldots, m_{\max}^{(k)}\}$ in polynomial time such that $\ALG(m)$ is within a factor $(1 + \bigO(\eps))$ of the optimal assignment given the same set of bags. 
This is formalized in the next lemma. 

\begin{restatable}{lemma}{lemscDPmain} \label{lem:sc_DP_main}
    Let $\guessbags[k]$, $m_{\max}^{(k)}$ and the job-to-bag assignment of~$J_k\cup J_{k-1}$ to $\guessbags[k]$ be guesses as defined. 
    Further, let $\hat\calC$ satisfy \eqref{eq:sc_dp_next} and suppose that the bag sizes in $\guessbags[k-1]$ are given by $(\hat s_\ell)_{\ell \in L_{k-1}}$. 
    Let $\bags'\cup\{B_{\hat S}\}$ be a partition of the jobs $\bigcup_{k'=1}^{k - 2} J_{k'}$ into $M-\noBags-|\hat{\bags}_{k}|-|\hat{\bags}_{k-1}|+1$ bags such that 
    \begin{itemize}
        \item for each bag $B\in\bags'$ we have that $p(B)\le\epsilon\cdot\big(\frac{1}{\epsilon}\big)^{3k}$, 
        \item $p(B_{\hat S})\ge \hat S$, 
        \item $p(\bags'):=\sum_{B\in\bags'}p(B)\ge(1+\epsilon)^{-1}T$. 
    \end{itemize}
    Suppose that each bag $B\in\hat{\bags}_{k}\cup\hat{\bags}_{k-1}$ satisfies $p(B)\in[(1+\epsilon)^{\ell(B)},(1+\epsilon)^{\ell(B)+1})$ and let $\guessbags[L]$ contain $\noBags$ many large bags of size at least $\big(\frac1\eps\big)^{3k + 3}$.
    There is a polynomial-time algorithm that either asserts that our guess is incorrect and cannot be combined with $\hat\calC$ or that computes a vector $\left(\ALG(m)\right)_{m=m_{\min}^{(k)}}^{m_{\max}^{(k)}}$ such that $\opt(\bags_{L}\cup\hat{\bags}_{k}\cup\hat{\bags}_{k-1}\cup\bags',m)\in[(1+\epsilon)^{-5}\ALG(m),(1+\epsilon)\ALG(m))$
    holds for each $m\in\{m_{\min}^{(k)},...,m_{\max}^{(k)}\}$ for which
    $\opt(\bags_{L}\cup\hat{\bags}_{k}\cup\hat{\bags}_{k-1}\cup\bags',m)\ge(1+\epsilon)^{-1}(\frac{1}{\epsilon})^{3k}$.
    
    Further, we can find the best $\hat\calC$ that satisfies \eqref{eq:sc_dp_next} and can be combined with our guess in polynomial time.
\end{restatable}

To compute the final solution, we combine the correct initial guesses with the solution stored in the DP cell corresponding to the remaining subproblem. This yields a global solution to the original problem.
In order to prove the correctness of our DP, we observe that for each $k \in [\kmax-1]$ there is a \emph{special} DP cell for which $k$ is the first parameter and whose other parameters correspond to the optimal solution (e.g., the assignment of jobs in $J_k$ to bags in $\bags^*_{k+1}$). We then prove by induction that, for each $k \in [\kmax-1]$, the solution stored in the corresponding special DP cell yields a profit that is similar to the optimal profit restricted to scenarios with $m\in \{1,...,m_{\max}^{(k)}\}$ machines, using \Cref{lem:sc_DP_main}.

\section{Conclusion}
In this paper, we continue the recent line of research on scheduling with uncertainty in the machine environment~\cite{BalkanskiOSW22,EberleHMNSS23,SteinZ20} by considering a stochastic machine environment in which the number of identical parallel machines is only known in terms of a distribution and the actual number is revealed once the jobs are assigned to bags which cannot be split anymore. 
Interestingly, we present polynomial time approximation schemes for minimizing the makespan as well as maximizing the minimum machine load, which matches their respective deterministic counterparts from the perspective of approximation algorithms.
We believe that our insights open up many interesting questions for future research such as extending the current model to the setting with uniformly related machines in which the uncertainty is modeled in terms of machine speeds as done in~\cite{EberleHMNSS23} from a robustness point-of-view or to the setting with different (job-based) objectives such as sum of weighted completion times.

\bibliographystyle{alpha}
\newpage
\bibliography{ref}

\newpage
\appendix
\section{Proofs for \Cref{sec:makespan}}\label{app:makespan}
\subsection{Proof of \Cref{lem:largest-bag-size}} 
\lemmakespanlargestbagsize*

\begin{proof}
    Consider a solution $\bags$ with $p(B) > 4 C$ for some  $B \in \bags$. We can lower bound the makespan in each scenario by $p(B)$. Hence, the overall objective function is at least 
    \[
        \sum_{m=1}^M q_m p(B) = p(B) > 4 C = 4 \sum_{m=1}^M q_m \max \Big\{ \max_{j \in J} p_j, \, \frac1m \sum_{j \in J} p_j \Big\} \, .
    \]
    The analysis of Graham's list scheduling algorithm~\cite{Graham1966} tells us that 
    \begin{equation}\label{eq:T-is-lb-on-opt}
        \max \Big\{ \max_{j \in J} p_j, \, \frac1m \sum_{j \in J} p_j \Big\} \leq \opt(J,m)  \, .        
    \end{equation}
    By using the set of bags $\bags'$ as designed by~\cite{SteinZ20}, we know that 
    \[
        \opt(\bags',m) \leq \Big(\frac53 + \bar \eps \Big) \opt(J,m) \leq 4 \max \Big\{ \max_{j \in J} p_j, \, \frac1m \sum_{j \in J} p_j \Big\} 
    \]
    holds for each scenario~$m$, where the first inequality follows from Theorem~3.10 in~\cite{SteinZ20} for $\bar \eps \in \big(0, \frac13 \big) $.  
    Thus, $\opt \leq 4  \sum_{m=1}^M q_m \max \big\{ \max_{j \in J} p_j, \, \frac1m \sum_{j \in J} p_j \big\}$. Therefore, $\bags$ cannot be an optimal solution as required. 
\end{proof}

\subsection{Proof of \Cref{lem:feasible-guess_packed-well}}
\lemmakespanfeasibleguesspackswell*

\begin{proof}
    We first argue that the total number of bags is at most~$M$. To this end, we observe that $p(B) < \eps C$ for $B \in \optbags\setminus\optbags[R]$. Hence, $|\optbags \setminus \optbags[R]| \geq \Big\lceil \frac{\sum_{B \in \optbags\setminus\optbags[R]} p(B) }{\eps T } \Big\rceil = \hat M_{\sand}$. Any regular bag~$B \in \optbags[R]$ is counted exactly once when defining the numbers~$\hat M_\ell$. Thus, $\sum_{\ell \in \calL} \hat  M_\ell +\hat M_{\sand} = |\optbags| \leq M$. 

    For a regular bag~$B \in \optbags[R]$, we observe that $p(B) \in [\eps^2 C, 4C)$ by our assumption on the maximal size of the bags in $\optbags$ and the definition of regular bags. Hence, $\hat \guess$ is indeed a possible guess of the algorithm. 

    It remains to argue that a set of bags $\bags$ with $\hat M_\ell$ bags of size $(1+\eps)^{\ell + 1}$ for each~$\ell \in \calL$ and $\hat M_{\sand}$ bags of size~$(1+\eps)\eps C$ can indeed pack all jobs. 
    To this end, we observe that any job assigned to a regular bag $B \in \optbags[R]$ with $p(B) \in \big[ (1+\eps)^{\ell}, (1+\eps)^{\ell+1}\big)$ can be assigned to a bag of size $(1+\eps)^{\ell+1}$, and, by construction, there are enough such bags. For the remaining jobs, we observe that 
    \[
        \hat M_{\sand} \eps C \geq \frac{\sum_{B \notin \optbags[R]} p(B_i) }{\eps C } \eps C= \sum_{B \notin \optbags[R]} p(B),
    \]
    that is, the total volume of the sand bags scaled down by a factor $(1+\eps)$ is at least the total volume of $\optbags\setminus\optbags[R]$. This observation allows us to define a fractional packing of the not yet packed jobs in the sand bags as follows: We fix an arbitrary order of those jobs and assign jobs in this order to the first sand bag until packing the next job would increase the total packed size beyond $\eps C$. We assign a fraction of this job to the first bag such that the total packed size is exactly $\eps C$ and treat the remaining fraction as first job for the second bag. Continuing inductively in this fashion, we can pack all jobs and pack at most $\hat M_{\sand}-1$ jobs fractionally.
    
    Of course, a feasible packing of jobs to bags needs to be integral. To this end, we completely pack each job in the bag where its first part was packed by the above fractional packing procedure. By definition, each job packed in $B \notin \optbags[R]$ has size at most $\eps^2 C$. Thus, the total size assigned to a sand bag is at most $\eps C + \eps^2 C = (1+ \eps) \eps C$, which is exactly the size of a sand bag created by the algorithm. Therefore, all jobs can feasibly be packed in a set of bags whose sizes are given by $\hat \guess$. 
\end{proof}

\subsection{Proof of \Cref{lem:good-guess}}

\lemmakespangoodguess*

\begin{proof} For simplicity, let $\opt(\hat \guess,m)$ denote $\opt((\hat M_\ell)_{\ell \in \calL\cup\{\sand\}}, m)$.
    We show the following statement: For each scenario $m$, 
    \[
        \opt(\hat \guess, m) \leq (1+\eps) \opt(\optbags,m) + 4 \eps C.
    \]     
    To this end, fix a scenario $m$. We assign the $\hat M_\ell$ bags of size $(1+\eps)^{\ell+1}$ to the same machines as the $\hat M_\ell$ regular bags with sizes in~$\big[ (1+\eps)^\ell, (1+\eps)^{\ell+1}\big)$. Next, we schedule the sand bags using list scheduling, i.e., we consider them in some fixed order and assign the next bag in this order to a machine that currently has the smallest total load. Denote the makespan of this assignment by $z_m$. Clearly, $\opt(\hat \guess, m) \leq z_m$. 

    Next, we argue that $z_m \leq (1+\eps)^2(1 + 2 \eps)\opt(\optbags,m)$ distinguishing two cases based on whether a regular or a sand bag determines the makespan. Consider the case where a regular bag of size $(1+\eps)^{\ell+1}$ determines the makespan~$z_m$ on some machine~$i$. This implies that no sand bag is scheduled on machine~$i$ as they are assigned last. For each~$\ell \in \calL$, the number of bags of size $(1+\eps)^{\ell+1}$ on machine~$i$ equals the number of regular bags with sizes in $\big[ (1+\eps)^\ell, (1+\eps)^{\ell+1}\big)$. Since the former bags are at most a factor~$(1+\eps)$ larger than the latter bags, we obtain $z_m \leq (1+\eps)\opt(\optbags,m)$ in this case. 
    
    Suppose now that some sand bag of size $(1+\eps)\eps C$ determines the makespan. Similar to the classical argumentation on list scheduling by \cite{Graham1966}, we can upper bound the starting time of this bag by 
    $\frac 1m \Big( \sum_{\ell \in \calL} \hat M_\ell (1+\eps)^{\ell+1} + \hat M_{\sand} (1+\eps)\eps C  \Big)$, which implies 
    \begin{align*}
        z_m & \leq \frac 1m \sum_{\ell \in \calL} \hat M_\ell (1+\eps)^{\ell+1} + \frac1m \hat M_{\sand} (1+\eps)\eps C + (1+ \eps)\eps C     \\
        & = \frac{1+\eps} m \sum_{B \in \optbags[R]}  p(B) + \frac1m \hat M_{\sand} \cdot (1+\eps) \eps C + (1+ \eps)\eps C, \shortintertext{where we used that we increase the total volume packed in regular bags at most by a factor~$(1+\eps)$. With the definition of $\hat M_{\sand}$, we get}
        z_m & \leq \frac{1+\eps} m \sum_{B \in \optbags[R]}  p(B) +  \frac1m \Bigg\lceil \frac{\sum_{B \in \optbags\setminus \optbags[R]} p(B) }{\eps C } \Bigg \rceil (1+\eps) \eps C    + (1+ \eps)\eps C \\
        & \leq \frac{1+\eps} m \sum_{B \in \optbags}  p(B) + 2 (1+\eps) \eps C \leq (1+\eps) \opt(\optbags,m) + 4\eps C.
    \end{align*}
    In the last line we used that $\frac1m \sum_{B \in \optbags} p(B) $ is a lower bound on $\opt(\optbags,m)$ and~$\eps \leq 1$. 

    Using $\sum_{m=1}^M q_m = 1$ and that $C$ is a lower bound on $\opt$ by \eqref{eq:T-is-lb-on-opt}, we conclude 
    \[
        \sum_{m=1}^M q_m \opt(\hat \guess,m) \leq \sum_{m=1}^M q_m  z_m \leq  (1+5\eps)\opt.
    \] 
    
\end{proof}

\subsection{Proof of \Cref{lem:approx-ratio}}

\lemmakespanapproxratio*

\begin{proof}
    As the algorithm returns the best guess, it is sufficient to give a guess that achieves the bound stated in the lemma. 
    The PTAS by \cite{HochbaumS87} guarantees that the value $z((\hat M_\ell)_{\ell \in \calL\cup\{\sand\}}, m)$ used for evaluating the guessed vector $(\hat M_\ell)_{\ell \in \calL\cup\{\sand\}}$ as defined above is upper-bounded by $(1+\eps) \opt(\hat \guess, m)$. With \Cref{lem:good-guess},  we get 
    \[
        \sum_{m =1}^M q_m z(\hat \guess, m) \leq (1+\eps) (1+5\eps)\opt. 
    \]    
    \Cref{cor:HS_bin-packing-PTAS} guarantees that the PTAS by \cite{HochbaumS88} packs all jobs exceeding the bag sizes in~$(\hat M_\ell)_{\ell \in \calL\cup\{\sand\}}$ by at  most a multiplicative factor $(1 + \eps)$. This in turn increases the actually achievable makespan in each scenario by at most a factor~$(1+ \eps)$, which  concludes the proof. 
    
\end{proof}

\subsection{Proof of \Cref{lem:running-time}}

\lemmakespanrunningtime*

\begin{proof}
    The algorithm guesses the sizes of regular bags up to a factor $(1+\eps)$, while all sand bags have a fixed size $\eps T$. There are at most $\lceil \log_{1+\eps} \frac4 \eps \rceil = \bigO\big(\frac1{\eps^2}\big)$ distinct bag sizes in $(\eps C, 4 C]$. As there are at most $M$ bags of each type, the number of possible guesses is bounded by $M^{\bigO(1/\eps^2)}$. For each guess, we check for feasibility using the PTAS by Hochbaum and Shmoys~\cite{HochbaumS88}, which runs in time $\bigO\big( M \cdot n^{10/\eps^2 + 3}\big)$.     
    For each guess and each scenario, we calculate an upper bound on the makespan using the PTAS by Hochbaum and Shmoys~\cite{HochbaumS87}, which runs in time $\bigO\big( \big( \frac n {\eps}\big)^{1/\eps^2} \big)$. Combining, we obtain a running time of 
    \[ 
        \bigO\Big(M^{\bigO(1/\eps^2)} \Big(M \cdot n^{10/\eps^2 + 3} + M  \cdot \Big( \frac n {\eps}\Big)^{1/\eps^2} \Big) \Big), 
    \]
    which concludes the proof as~$M < n$. 
    
\end{proof}

\section{Proofs for \Cref{sec:santaclaus}} \label{app:santa}

\subsection{Proof of \Cref{lem:rounding}}

In this section, we prove \Cref{lem:rounding}. We start by making some assumpions on the instance.
By decreasing~$\eps$ if necessary and by rescaling the instance, we can assume that $\frac1\eps \in \mathbb Z_{\geq 2}$ and $p_j \in \mathbb N$ for each~$j \in J$.
Let $d$ be the smallest integer such that $\left( \frac{n}{\epsilon} \right)^d > \sum_{j \in J} p_j$. We observe that the optimal value $\opt(\optbags,m)$ for each scenario~$m$ is in $[1,\left( \frac{n}{\epsilon} \right)^d]$.

In the following, we partition the interval $[1,\left( \frac{n}{\epsilon} \right)^d]$ into a polynomial number of intervals such that any two values in the same interval are within a factor of $\left( \frac{n}{\epsilon} \right)^{\frac{1}{\epsilon}}$ and such that two values in any two consecutive intervals have a multiplicative gap of at least $\left( \frac{n}{\epsilon} \right)^3$.
For each interval we will construct an instance with the desired structure as described in \Cref{lem:rounding}.
Then, we will show how to solve the original problem by combining the solutions to the smaller instances with a dynamic program while losing at most a factor of $1+\bigO(\epsilon)$ in the approximation factor.

Let $K$ be the smallest integer such that $\left( \frac{n}{\epsilon} \right)^{3K + \frac{K-1}{\epsilon}} \geq \left( \frac{n}{\epsilon} \right)^d$.
For each index $k \in \{0,1,\dots,K\}$ and each offset $a \in \big\{0,1,\dots,\frac{1}{\epsilon} + 3\big\}$, let $\Tilde{I}^a_k : = \Big[ \left( \frac{n}{\epsilon} \right)^{3k + \frac{k-1}{\epsilon} + a}, \left( \frac{n}{\epsilon} \right)^{3k + \frac{k}{\epsilon} + a}\Big)$.
For each $a$, define a collection of intervals $\Tilde{\mathcal{I}}_a = \{ \Tilde{I}^a_k : k \in\{0,1,\dots,K\}\}$, and let $\mathcal{M}_a = \{m \in [M] : \opt(\optbags,m) \in \bigcup_{\Tilde{I} \in  \Tilde{\mathcal{I}}_a} \Tilde{I} \}$ denote the scenarios with $\opt(\optbags,m) \in \bigcup_{\Tilde{I} \in  \Tilde{\mathcal{I}}_a} \Tilde{I}$.

The next lemma shows that there is a choice of an offset $a$ that allows us to ignore scenarios not in $\mathcal{M}_a$ and still get good approximations because the expected contribution to $\opt$ from those scenarios is small.

\begin{lemma}
    \label{lem:ingnore_eps_scenarios}
    There is an $a \in \{0,1,\dots,\frac{1}{\epsilon} +3\}$ such that
    \[
        \sum_{m \in \mathcal{M}_a} q_m \opt(\optbags,m)
            \geq \frac{1}{1+4 \epsilon} \opt.
    \]
\end{lemma}
\begin{proof}
    Choose $a$ uniformly at random from $\{0,1,\dots,\frac{1}{\epsilon}+3\}$, then the probability that a certain value for~$a$ is chosen is given by $p_a =\frac{1}{4+\frac{1}{\epsilon}} = \frac{\epsilon}{1+4\epsilon}$. 
    Also, notice that for any $m$ there are $\frac{1}{\epsilon}$ values of $a \in \{0,1,\dots,\frac{1}{\epsilon}+3\}$ such that $m \in \mathcal{M}_a$.
    Hence,
    \begin{align*}
        E \left[ \sum_{m \in \mathcal{M}_a} q_m \opt(\optbags,m) \right]
            & = \sum_{a} p_a \sum_{m \in \mathcal{M}_a} q_m \opt(\optbags,m) \\
            & = \frac{\epsilon}{1+4\epsilon} \sum_{a} \sum_{m \in \mathcal{M}_a} q_m \opt(\optbags,m) \\
            & = \frac{\epsilon}{1+4\epsilon} \sum_{m \in [M]} \frac{1}{\epsilon} \cdot q_m \opt(\optbags,m) \\
            & = \frac{1}{1+4\epsilon} \sum_{m \in [M]} q_m \opt(\optbags,m) \\
            & = \frac{1}{1+4\epsilon} \opt. 
    \end{align*}
    Thus, there is an $a$ such that $\sum_{m \in \mathcal{M}_a} q_m \opt(\optbags,m) \geq \frac{1}{1+4\epsilon} \opt$.  
\end{proof}

Because there is a polynomial number of possible values for the offset $a$, by enumeration, we can assume that the chosen value~$a$ satisfies \Cref{lem:ingnore_eps_scenarios}.
Thus, to simplify notation, we drop the subscript from the collection of intervals $\Tilde{\mathcal{I}}_a$ and the superscript from $\Tilde{I}^a_k$.
As a consequence of \Cref{lem:ingnore_eps_scenarios} we also assume that the scenarios with $m \notin \mathcal{M}_a$ do not occur for the remainder of this section.
This is formalized in \Cref{cor:ingnore_eps_scenarios}.

\begin{corollary}
    \label{cor:ingnore_eps_scenarios}
    By losing a factor of at most $1 + 4\epsilon$ in the approximation ratio, we can assume that $q_m = 0$ for every $m \notin \mathcal{M}_a$.
\end{corollary}

For each interval $\Tilde{I}_k \in \Tilde{\mathcal{I}}$, it is useful to consider the super-interval that also contains the gap $\Big[\left( \frac{n}{\epsilon} \right)^{3k+\frac{k-1}{\epsilon}+a-3},\left( \frac{n}{\epsilon} \right)^{3k+\frac{k-1}{\epsilon}+a}\Big)$.
We denote by $\Tilde{I}^+_k = \Big[\left( \frac{n}{\epsilon} \right)^{3k + \frac{k-1}{\epsilon} + a - 3}, \left( \frac{n}{\epsilon} \right)^{3k + \frac{k}{\epsilon} + a}\Big)$ the \emph{extended} interval of $\Tilde{I}_k$, and we denote by $J_k = \{j \in J : p_j \in \Tilde{I}^+_k\}$ the set of jobs whose processing time fall in the extended interval $\Tilde{I}^+_k$.
Let $\Tilde{\mathcal{I}}^+$ denote the collection of extended intervals and notice that, similar to the intervals in $\Tilde{\mathcal{I}}$, the extended intervals are pair-wise disjoint.

In the following \Cref{lem:bags_restricted_to_intervals} we further develop the assumption on the instance.
This lemma shows that that we can assume that bags only contain jobs from the same extended interval.

\begin{lemma}
    \label{lem:bags_restricted_to_intervals}
    By losing a factor of at most $(1 + \frac{\epsilon}{1-\epsilon})^2$ in the approximation ratio, we can assume that no bag contains jobs $j,j'$ whose processing times $p_j, p_{j'}$ are in distinct extended intervals.
\end{lemma}
\begin{proof}
    For a set of bags $\bags$, let $\bags_k = \{ B \in \bags : B \cap J_k \neq \emptyset \}$ denote the set of bags containing a job from $J_k$.
    For convenience, define $\bags_k = \emptyset$ if $k \notin \{0,\dots,K\}$.

    We modify $\optbags$ with the following iterative process.
    Start by defining $\bags := \optbags$, then start the iterative process.
    For $k = K,\dots,0$ do: let $J'$ be the set of jobs contained in bags in $\bags_k$ that are also contained in $\bigcup_{k'<k} J_{k'}$.
    If there is a $k'' < k$ such that $\bags_{k''} \neq \emptyset$, then move every job in $J'$ into a bag $B \in \bags_{k''}$; otherwise exclude every job in $J'$ from the bags in $\bags_k$ and end the iterative process.
    Return~$\bags$.

    Observe the procedure might not place some jobs from $J$ to any bags, as they were excluded in the last iteration.
    We argue why we may assume that no jobs were excluded at the end without loss of generality.
    Let $j'$ be the job with greatest processing time among the excluded jobs and let $J_{\bar k}$ be such that $j' \in J_{\bar k}$.
    Notice that every bag in $\bags$ contains at least one job from $\bigcup_{k > \bar k} J_{k}$ and that, together with \Cref{cor:ingnore_eps_scenarios}, we have that for every scenario $m$ with $q_m > 0$
    \begin{align*}
        \sum_{j \in \bigcup_{k'\leq \bar k} J_k} j
            & \leq n \left( \frac{n}{\epsilon} \right)^{3 \bar k +\frac{\bar k}{\epsilon} + a}
                \leq \left( \frac{n}{\epsilon} \right)^{3 \bar k +\frac{\bar k}{\epsilon} + a + 1}
                \leq \epsilon \cdot \opt(\optbags,m),
    \end{align*}
    which implies $\opt(\bags,m) \geq (1-\epsilon) \opt(\optbags,m)$ for scenarios $m$ with $q_m > 0$.
    Therefore, by enumerating the collections $\cup_{k>\bar k} J_k$, we can assume that there are no jobs from $\bigcup_{k \leq \bar k} J_{k}$, obtain a solution, and then arbitrarily assign the jobs in $\bigcup_{k \leq \bar k} J_{k}$ to bags. We can show that this procedure costs us at most a factor $1+\frac{\epsilon}{1-\epsilon}$ in the approximation ratio.
    
    For a bag $B \in \optbags$, let $\bags(B)$ denote its corresponding modified bag in $\bags$.
    Notice at any iteration of the procedure, a bag never have its jobs with greatest processing time removed:
    Consider an iteration $k$ and a bag $B$, and let $j$ be the job with greatest processing time that is currently in $B$.
    Suppose, for a contradiction that $j$ is removed from $B$ in this iteration.
    Then $B \in \mathcal{B}_k$, and we have that $B$ contains a job $j' \in J_k$, i.e., $p_{j'} \in \Tilde{I}^+_k$.
    But since $j$ was removed, it means that $j \in J_{k'}$ for some $k' < k$, and this implies that $p_{j} < p_{j'}$, which is a contradiction.
    
    Thus, for any bag $B$, we have that the largest processing time among initial jobs of $B$ is a lower bound for the total processing time of jobs in $B$ after the procedure, i.e., $\max_{j \in B} p_j \leq \sum_{j \in \bags(B)} p_j$.
    Also, because the iterative process maintains the invariant that at iteration $k = \bar k$, bags in $\bags_{k}$ contains only jobs from $J_{k}$ for every $k>\bar k$, it follows that the procedure returns a collection of bags $\bags$ of which no bag contains jobs $j$ and $j'$ whose processing times $p_j$ and $p_{j'}$ are in distinct extended~intervals.
    
    Let $\opt(\optbags,m)$ be such that $q_m >0$ and consider the schedule for $\optbags$ that achieves $\opt(\optbags,m)$.
    By \Cref{cor:ingnore_eps_scenarios}, there is a $k$ such that $\opt(\optbags,m) \in \Tilde{I}_{k}$.
    Suppose, for the sake of contradiction, that one of the $m$ machines contains only bags with jobs of processing time strictly smaller than $\left( \frac{n}{\epsilon} \right)^{3{k}+\frac{{k}-1}{\epsilon}+a-1}$, then the load on that machine is at most $n \cdot \left( \frac{n}{\epsilon} \right)^{3{k}+\frac{{k}-1}{\epsilon}+a-1} < \left( \frac{n}{\epsilon} \right)^{3{k}+\frac{{k}-1}{\epsilon}+a}$, which is a contradiction as this implies that $\opt(\optbags,m) \notin \Tilde{I}_{k}$.
    We~conclude that every machine is scheduled at least a bag containing a job whose processing time is at least~$\left( \frac{n}{\epsilon} \right)^{3{k}+\frac{{k}-1}{\epsilon}+a-1}$.

    Let us consider by how much the scheduled processing time of a machine~$i$ can decrease when we replace the bags scheduled on $i$ by $\opt(\optbags,m)$ by the corresponding bags in $\bags$.
    Let $\optbags(i) = \{B \in \optbags : B \rightarrow i\}$ denote the collection of bags scheduled on $i$ in the schedule achieving $\opt(\optbags,m)$, and let $\Tilde{\optbags}(i) = \{B \in \optbags(i) : \sum_{j \in B} p_j > \sum_{j \in \bags(B)} p_j \}$ denote the bags scheduled on $i$ whose size was decreased by the procedure.
    Then the decrease can be at most
    \begin{align*}
        \sum_{B \in \optbags(i)} \left( \sum_{j \in B} p_j - \sum_{j \in \bags(B)} p_j \right)
            & \leq \sum_{B \in \Tilde{\optbags}(i)} \left( \sum_{j \in B} p_j - \sum_{j \in \bags(B)} p_j \right) \\
            & \leq \sum_{j \in J_{k'} : k' < k} p_j \\
            & < n \left( \frac{n}{\epsilon} \right)^{3k+\frac{k-1}{\epsilon}+a-3} \\
            & \leq n \left( \frac{\epsilon}{n} \right) \opt(\optbags,m) \\
            & = \epsilon \cdot \opt(\optbags,m).
    \end{align*}
    Since the bound for the decrease is valid for every machine in every scenario $m$ with $q_m > 0$, it follows that
    \begin{align*}
        \sum_{m } q_m \opt(\bags,m)
            & = \sum_{m : q_m > 0} q_m \opt(\bags,m) \\
            & \geq \sum_{m : q_m > 0} q_m \opt(\optbags,m) (1-\epsilon) \\
            & = \sum_{m } q_m \opt(\optbags,m) (1-\epsilon) \\
            & = (1-\epsilon) \opt \\
            & = \frac{1}{1+\frac{\epsilon}{1-\epsilon}} \opt.
    \end{align*}
This completes the proof.  
\end{proof}

Through the remainder of this section, we assume that no bag contains jogs $j,j'$ whose processing times are in distinct extended interval.

In \Cref{lem:bags_restricted_to_intervals} we showed a result that restrict the kind of jobs assigned to a bag.
Then, in \Cref{lem:bags_in_interval_dominates_opt_in_interval} we prove a similar result to that restricts the kind of bags assigned to machines with minimum load.

\begin{lemma}
    \label{lem:bags_in_interval_dominates_opt_in_interval}
    By losing a factor of at most $1+\frac{\epsilon}{1-\epsilon}$, we may assume that for each scenario $m$, each machine with minimum load is assigned only bags with jobs from $J_k$, where $k$ is the index of the interval $\Tilde{I}_k$ with $\opt(\optbags,m) \in \Tilde{I}_k$.
\end{lemma}
\begin{proof}
    As bags only contain jobs whose processing times are in the same extended interval and since $\opt(\optbags,m) \in \Tilde{I}_k$, it follows that no machine with minimum load processes a bag with a job from $\bigcup_{k'>k} J_{k'}$.
    
    Recall that the total processing time of the jobs in $J_{k'<k}$ is at most
    \begin{align*}
        n \left( \frac{n}{\epsilon} \right)^{3k + \frac{k-1}{\epsilon}+a-3}
            \leq n \left(\frac{\epsilon}{n}\right) \opt(\optbags,m)
            = \epsilon \cdot \opt(\optbags,m).
    \end{align*}
    Therefore, removing every bag containing jobs not in $J_k$ from the machines with minimum load reduces their total processing time by at most $\epsilon$. 
\end{proof}

Through the remainder of this section, we assume that, for each scenario $m$, the machines with minimum load is assigned only bags with jobs from $J_k$, where $k$ is the index of the interval $\Tilde{I}_k$ with $\opt(\optbags,m) \in \Tilde{I}_k$.

For each $k$, we define a small interval within the gaps between the intervals in $\Tilde{\mathcal{I}}$.
Let the \emph{head gap} of the extended interval $\Tilde{I}^+_{k}$ be defined as $\Tilde{G}_k = \Big[\left( \frac{n}{\epsilon}\right)^{3k+\frac{k-1}{\epsilon}+a-3}, \left( \frac{n}{\epsilon}\right)^{3k+\frac{k-1}{\epsilon}+a-2}\Big)$.
Notice that any two values in~$\Tilde{G}_k$ are within a factor of $\left( \frac{n}{\epsilon} \right)$ of each other, and that $\Tilde{G}_k \subsetneq \Tilde{I}^+_{k} \setminus \Tilde{I}_{k}$ is a subinterval of the gap between the intervals $\Tilde{I}_{k-1}$ and $\Tilde{I}_{k}$.
In \Cref{lem:headgap_empty} we show that we may ignore jobs whose processing times fall in the head gaps.

\begin{lemma}
    \label{lem:headgap_empty}
    By losing a factor of at most $1+\frac{\epsilon}{1-\epsilon}$, we may assume that there are no jobs whose processing time are in $\bigcup_{k} \Tilde{G}_k$.
\end{lemma}
\begin{proof}
    Consider a scenario $m$ and let $\Tilde{I}_k$ be such that $\opt(\optbags,m) \in \Tilde{I}_k$.
    By \Cref{lem:bags_in_interval_dominates_opt_in_interval}, a machine with minimum load in the schedule achieving $\opt(\optbags,m)$ only processes bags with jobs from $J_k$.
    Thus, if this machine processes a bag with a job whose processing time is from a head gap, it must be the head gap $\Tilde{G}_k$.
    But the total processing time of jobs with processing times in $\Tilde{G}_{k}$ is at most
    \begin{align*}
        n \left( \frac{n}{\epsilon} \right)^{3k + \frac{k-1}{\epsilon}+a-2}
            < n \left(\frac{\epsilon}{n}\right) \opt(\optbags,m)
            = \epsilon \cdot \opt(\optbags,m).
    \end{align*}
    Thus, removing the jobs with processing times in the gap $\bigcup_{k'\leq k} \Tilde{G}_{k'}$ decreases the load on a machine with minimum load in each scenario~$m$ by at most $\epsilon$. 
\end{proof}

Through the remainder of this section, we assume that there are no jobs whose processing times are in $\bigcup_{k} \Tilde{G}_k$.

\begin{corollary}
    \label{cor:headgap_empty}
    The sum of the processing times of jobs in $J_{k-1}$ is strictly less than $\epsilon$ times the processing time of any job in $\bigcup_{k'\geq k} J_{k'}$, for every $k$ such that $\bigcup_{k'\geq k} J_{k'} \neq \emptyset$.
\end{corollary}
\begin{proof}
    As there are no jobs whose processing times lie in some head gap, it follows that
    \begin{align*}
        \sum_{j \in J_{k-1}} p_j
            & \leq n \left( \frac{n}{\epsilon} \right)^{3k + \frac{k-1}{\epsilon} + a - 3}
                = \epsilon \left( \frac{n}{\epsilon} \right)^{3k + \frac{k-1}{\epsilon} + a - 2}
                < \epsilon \cdot \min_{j \in \bigcup_{k'\geq k} J_{k'}} p_j.
    \end{align*}
\end{proof}

\begin{lemma}
    \label{lemma:m_m1_machines_only_jobs_from_J_i}
    Let $m_1$ and $m_2$ be the smallest and greatest number of machines such that $\opt(\optbags,m_1), \opt(\optbags,m_2) \in \Tilde{I}_k$ for some $k$ such that $\bigcup_{k'\geq k} J_{k'} \neq \emptyset$.
    By losing at most $1 + \frac{\epsilon}{1-\epsilon}$ we may assume that for each scenario $m \in [m_1,m_2]$, there are at least $m - m_1 + 1$ machines containing only bags with jobs from $J_k$.
\end{lemma}
\begin{proof}
    Because $m_1$ is the smallest number of machines such that $\opt(\optbags,m_1) \in \Tilde{I}_k$ and every machine with the smallest load in this scenario contain only bags with jobs from $J_k$, it follows that the number of bags with jobs from $\bigcup_{k' > k} J_{k'}$ is strictly less than $m_1$.
    Therefore for each scenario $m \in [m_1,m_2]$, there are at least $m - m_1 + 1$ machines without bags with jobs from $\bigcup_{k' > k} J_{k'}$.
    Let $O(m)$ denote the value obtained from the schedule achieving $\opt(\optbags,m)$ by removing the bags with jobs from $\bigcup_{k'<k} J_{k'}$.
    From \Cref{cor:headgap_empty}, we have that the sum of the processing times of the jobs in $J_{k-1}$ is strictly less than $\epsilon$ times the processing time of any job in $\bigcup_{k'\geq k} J_{k'}$, therefore,
    \begin{align*}
        O(m) &\geq \frac{1}{1+\frac{\epsilon}{1-\epsilon}} \opt(\optbags,m).
    \end{align*}
    
\end{proof}

Through the remainder of this section, we assume that, for each scenario $m \in [m_1,m_2]$, there are at least $m - m_1 + 1$ machines containing only bags with jobs from $J_k$, where $m_1$ and $m_2$ are the smallest and largest number of machines such that $\opt(\optbags,m_1), \opt(\optbags,m_2) \in \Tilde{I}_k$.

Recall that \Cref{lem:rounding} arises by showing that the general problem can be solved by solving subproblems restricted on each extended interval with a dynamic programming (DP) approach.
In \Cref{lem:rounding_aux1} we show that we may assume certain properties on each subproblem by losing an arbitrarily small factor in the approximation ratio.

\begin{lemma}
    \label{lem:rounding_aux1}
    Let $p_{\min}$ and $p_{\max}$ respectively be the smallest and the largest processing times among jobs in $J$.
    If $\frac{p_{\max}}{p_{\min}} \leq n^{c(\epsilon)}$ for some constant $c(\epsilon)$, then by losing a factor of at most $1+\epsilon$ in the approximation ratio, we can assume for each job $j \in J$ that $p_j = \lceil (1+\epsilon)^{\ell_j} \rceil$ for some $\ell_j \in \mathbb{N}_0$ and $p_j \in [1,n^{c(\epsilon)}]$, where $c(\epsilon)$ is some global constant.
\end{lemma}
\begin{proof}
    Since $\frac{p_{\max}}{p_{\min}} \leq n^{c(\epsilon)}$, by normalizing to new processing times $p'_j = \frac{p_j}{p_{\min}}$, every processing time now falls into the interval $[1,n^{c(\epsilon)}]$.
    Observe that the original and the normalized problem are equivalent.

    Suppose from now on that the processing times are normalized, and for each job $j \in J$ let $\ell_j$ be the largest integer such that $p'_j \geq (1+\epsilon)^{\ell_j}$, and define a new processing time $p''_j = \lceil (1+\epsilon)^{\ell_j} \rceil$.
    Observe that $p''_j \geq \frac{1}{1+\epsilon} p'_j$ for every job $j \in J$.
    Let $\opt''$ denote the optimum value with the new processing times.
    Consider the packing and schedule of the optimal solution $\optbags$ for the processing times~$p_j'$ for a scenario with $m$ machines, and let $\textsc{O}(m)$ denote the resulting objective function value that is obtained by replacing the processing times~$p_j'$ with~$p_j''$.
    Then,
    \begin{align*}
        \opt''
            & \geq \sum_{m=1}^M q_m \textsc{O}(m)
                \geq \sum_{m=1}^M q_m \frac{1}{1+\epsilon} \opt(\optbags,m)
                = \frac{1}{1+\epsilon} \opt.
    \end{align*}
    
\end{proof}

We are now ready to prove \Cref{lem:rounding}.

\lemrounding*
\begin{proof}
    Assume that there is a $\alpha$-approximation algorithm for instances where $p_j = \lceil (1+\epsilon)^{\ell_j} \rceil$ for some $\ell_j \in \mathbb{N}_0$ and $p_j \in [1,n^{c(\epsilon)}]$ for all jobs $j \in J$.
    Then by \Cref{lem:rounding_aux1}, this is a $(1 + \epsilon)\alpha$-approximation for the subproblems where the job set is restricted to $J_k$ for each $k$.

    Let $U_k = \bigcup_{k' \geq k} \Tilde{I}^+_{k'}$ denote the union of the last $K-k+1$ extended intervals of $\Tilde{\mathcal{I}}^+$, and let $H_k = \bigcup_{k' \geq k} J_{k'} = \{j \in J : p_j \in U_k\}$ denote the set of jobs whose processing times are in $U_k$.
    We define a DP where each cell $\mathcal{C}[k, m_{\max}, b]$ is specified by
    \begin{itemize}
        \item $k \in\{1,\dots,K\}$, specifying the interval $U_k$ and the subset of the jobs $H_k$;
        \item $b \in \{0,\dots, M\}$, counting the number of bags of $\optbags$ that contain jobs in $H_k$;
        \item $m_{\max} \in \{0,\dots, b\}$, denoting the maximal number of machines such that $\opt(\optbags,m) \in U_k$ for all $m\leq m_{\max}$.
    \end{itemize}
    For each $k$, as it is not clear how to compute the correct values $b$ and $m_{\max}$ a priori, we compute the cells $\mathcal{C}[k, m_{\max}, b]$ for every combination of $m_{\max}$ and $b$ with $0 \leq m_{\max} \leq b \leq M$.

    For a set of bags $\bags$ and a corresponding schedule on $m$ machines, let $l(\bags,m)$ denote the minimum machine load.
    Then we define the \emph{partial expected value} of $\bags$ and its corresponding schedule for up to $m_{\max}$ machines as
    \begin{align*}
        v_{m_{\max}}(\bags) := \sum_{m =1}^{m_{\max}} q_{m} \cdot l(\bags,m).
    \end{align*}
    
    For each cell $\mathcal{C}[k, m_{\max}, b]$, we compute and store a set of~$b$ bags $\bags[k, m_{\max}, b]$ together with a corresponding schedule for each scenario with $m \leq m_{\max}$ machines.
    Each such cell is computed by executing the $(1 + \epsilon)\alpha$-approximation to an instance restricted to the jobs in $J_k$ and merge the obtained solution with a suitable solution stored in some cell $\mathcal{C}[k+1, m_{\max}', b']$, for some $m_{\max}'$ and $b'$, that maximizes the partial expected value (we describe the details about the restricted instances and the merging later).
    As we compute the DP table, we will also ensure that each cell $\mathcal{C}[k, m_{\max}, b]$ has the following properties:
    \begin{itemize}
        \item If a bag in $\bags[k, m_{\max}, b]$ contains a job from $J_{k'}$, then every job in that bag is also from $J_{k'}$, and
        \item If $m_{\max}$ is the largest number of machines such that $\opt(\optbags,m_{\max}) \in U_k$ and $\optbags$ contains $b$ bags with jobs from $H_k$, then
        \begin{align*}
            v_{m_{\max}}(\bags[k, m_{\max}, b])
                & = \sum_{m =1}^{m_{\max}} q_m \cdot l(\bags[k, m_{\max}, b],m)
                    \geq \frac{1}{(1 + \epsilon)\alpha} \sum_{m =1}^{m_{\max}} q_m \opt(\optbags,m).
        \end{align*}
    \end{itemize}
    Therefore, by completing the DP table and computing $\mathcal{C}[1, M, M]$, we obtain a $(1 + \epsilon)\alpha$-approximate solution.

    We now describe how to compute the DP table.
    First, define
    \begin{align*}
        Q^t_s & :=
            \begin{cases}
                \sum_{m=s}^t q_m    & \text{ if there is an $m \in \{s,s+1,\cdots,t\}$ with $q_m > 0$,} \\
                1                   & \text{ otherwise.} \\
            \end{cases}
    \end{align*}
    We start by computing $\mathcal{C}[K, m_{\max}, b]$ for every $0 \leq m_{\max} \leq b \leq M$.
    If $J_{K} = \emptyset$, then $\bags[K, m_{\max}, b]$ is a collection of $b$ empty bags.
    Otherwise, let $p_{\max}$ and $p_{\min}$ be the largest and smallest processing times in $J_{K}$, and recall that by the definition of $\Tilde{I}^+_1$ we have that $\frac{p_{\max}}{p_{\min}} \leq n^{c(\epsilon)}$.
    Then, we can obtain, in polynomial time, a $(1 + \epsilon)\alpha$-approximate solution $\Tilde{\bags}(K, m_{\max}, b)$ and a corresponding schedule for the instance where the set of jobs $\Tilde{J} = J_{K}$, the processing times $\Tilde{p}_j = p_j$ for each job $j \in \Tilde{J}$, the maximal number of machines is $\Tilde{M} = b$ with probabilities $\Tilde{q}_m = \frac{q_m}{Q_{1}^{m_{\max}}}$ if $m \in [1,m_{\max}]$ and $\Tilde{q}_m = 0$ if $m \in (m_{\max},b]$.
    Notice that in this new instance, the number of bags and scenarios is~$b$, but the number of machines is actually at most $m_{\max}$ as the probability of having more than $m_{\max}$ machines is zero.
    
    Let $\Tilde{\optbags}(K, m_{\max}, b)$ denote an optimal set of bags for the constructed instance, let $\Tilde{O}(\Tilde{\optbags}(K, m_{\max}, b), m)$ denote the value of its corresponding optimal schedule for the scenario with $m$ machines, and let $\Tilde{O} = \sum_{m \leq b} \Tilde{q}_m \Tilde{O}(\Tilde{\optbags}(K, m_{\max}, b), m)$ denote the expected value of this optimal solution.
    Let $m_{\max}'$ be the greatest number of machines such that $\opt(\optbags, m_{\max}') \in \Tilde{I}^+_{K}$.
    If there is no such $m_{\max}'$, let $m_{\max}' = 0$.
    Let $b'$ be the number of bags in $\optbags$ with jobs from $J_{K}$ and recall that, by \Cref{lem:bags_restricted_to_intervals}, we assume that those bags contain only jobs from $J_{K}$.
    Further, \Cref{lem:bags_in_interval_dominates_opt_in_interval}, we assume that the machines with minimum load only contain bags with jobs from $J_{K}$, hence the expected value of the schedule that achieves $\opt(\optbags, m)$ on the scenarios where $m \in [1,m_{\max}']$ with probabilities $\Tilde{q}_m$ cannot be higher than the value of $\Tilde{O}$.
    This is,
    \begin{align*}
        \sum_{m = 1}^{m_{\max}'} \Tilde{q}_m \cdot \Tilde{O}(\Tilde{\optbags}(K, m_{\max}', b'), m)
            & \geq \sum_{m = 1}^{m_{\max}'} \Tilde{q}_m \cdot \opt(\optbags, m ).
    \end{align*}
    Therefore, for $m_{\max}'$ and $b'$ we have that
    \begin{align*}
        v_{m_{\max}'}(\Tilde{\bags}(K, m_{\max}, b))
            & = \sum_{m = 1}^{m_{\max}'} q_{m} \cdot l(\Tilde{\bags}(K, m_{\max}', b'), m) \\
            & = Q^{m_{\max}'}_{1} \sum_{m = 1}^{m_{\max}'} \Tilde{q}_{m} \cdot l(\Tilde{\bags}(K, m_{\max}', b'), m) \\
            & \geq Q^{m_{\max}'}_{1} \frac{1}{(1+\epsilon) \alpha} \sum_{m = 1}^{m_{\max}'} \Tilde{q}_{m} \cdot \Tilde{O}(\Tilde{\optbags}(K, m_{\max}', b'), m) \\
            & \geq Q^{m_{\max}'}_{1} \frac{1}{(1+\epsilon) \alpha} \sum_{m = 1}^{m_{\max}'} \Tilde{q}_m \cdot \opt(\optbags, m ) \\
            & = \frac{1}{(1+\epsilon) \alpha} \sum_{m = 1}^{m_{\max}'} q_m \cdot \opt(\optbags, m ),
    \end{align*}
    and that for every $m_{\max}$ and $b$ such that $0 \leq m_{\max} \leq b \leq M$ we have that the bags in $\bags[K, m_{\max}, b]$ are either empty or contain only jobs in $J_{K}$.
    Thus, for each $m_{\max}$ and $b$ such that $0 \leq m_{\max} \leq b \leq M$, we compute and store in~$\mathcal{C}[K,m_{\max},b]$ the bag set $\bags[K, m_{\max}, b] := \Tilde{\bags}(K, m_{\max}, b)$ together with a corresponding schedule for each scenario $m \leq m_{\max}$, and these cells satisfy the two properties of the DP table in polynomial time.

    Assume that for some $k$, every cell $\mathcal{C}[k',m_{\max}',b']$ with $k' \geq k$ and $0 \leq m_{\max}' \leq b' \leq M$ has been computed.
    We now describe how to compute the cell $\mathcal{C}[k-1,m_{\max},b]$ for $0 \leq m_{\max} \leq b \leq M$.
    If $J_{k-1} = \emptyset$ there are no new jobs to be considered, i.e. $H_{k-1} = H_k$, thus we copy the cell $\mathcal{C}[k,m_{\max},b]$ into $\mathcal{C}[k-1,m_{\max},b] $.
    Otherwise, we construct suitable candidate sets of bags $\mathcal{A}(m_{\max,1},b_1,m_{\max,2},b_2)$ with a corresponding schedule for each $m_{\max,1} \leq b_1 \leq M$ and $m_{\max,2} \leq b_2 \leq M$ among which we choose one with maximal partial expected value to be stored.
    The parameters $m_{\max,1}$ and $b_1$ indicate the cell $\mathcal{C}[k,m_{\max,1},b_1]$ from which we take the bags with jobs in $H_k$.
    For bags with jobs in $J_{k-1} = H_{k-1} \setminus H_{k}$, we obtain an $(1+\epsilon)\alpha$-approximate solution for an instance with $m_{\max,2} = m_{\max} - m_{\max,1}$ machines and $b_2 = b - b_1$ bags (similar to the ones for the cases $\mathcal{C}[K,m_{\max},b]$).
    
    As $J_{k-1} \neq \emptyset$, let $p_{\max}$ and $p_{\min}$ be the largest and smallest processing times in $J_{k-1}$, and recall that by the definition of $\Tilde{I}^+_{k-1}$ we have that $\frac{p_{\max}}{p_{\min}} \leq n^{c(\epsilon)}$.
    Then, for each $1 \leq m_{\max,2} \leq b_2 \leq M$, we can obtain, in polynomial time, a $(1 + \epsilon)\alpha$-approximate solution $\Tilde{\bags}(k-1, m_{\max,2}, b_2)$ and a corresponding schedule for the instance with job set $\Tilde{J} = J_{k-1}$, processing times $\Tilde{p}_j = p_j$ for $j \in \Tilde{J}$, and maximal number of machines $\Tilde{M} = b_2$ with probabilities $\Tilde{q}_{m} = \frac{q_{m_{\max}-m_{\max,2}+m}}{Q_{m_{\max}-m_{\max,2}+1}^{m_{\max}}}$ if $m \in [1,m_{\max,2}]$ and $\Tilde{q}_m = 0$ if $m \in (m_{\max,2}, b_2]$.
    Notice that in this new instance, the number of bags and scenarios is~$b_2$, but the number of machines is actually at most $m_{\max,2}$, as the probability of having more than $m_{\max,2}$ machines is zero.

    Let $m_{\max,1} = m_{\max} - m_{\max,2}$ and $b_1 = b - b_2$ and define the candidate set of bags $\mathcal{A}(m_{\max,1},b_1,m_{\max,2},b_2) = \bags[k,m_{\max,1},b_1] \cup \Tilde{\bags}(k-1, m_{\max,2}, b_2)$, and define its corresponding schedule as follows.
    \begin{itemize}
        \item If $m \in [1,m_{\max,1}]$, then schedule the bags in $\mathcal{A}(m_{\max,1},b_1,m_{\max,2},b_2) \cap \bags[k,m_{\max,1},b_1]$ identical to the schedule corresponding to $\bags[k,m_{\max,1},b_1]$ and schedule the bags in $\mathcal{A}(m_{\max,1},b_1,m_{\max,2},b_2) \cap \Tilde{\bags}(k-1, m_{\max,2}, b_2)$ arbitrarily.
        \item If $m \in (m_{\max,1},m_{\max}]$, then schedule the bags $\mathcal{A}(m_{\max,1},b_1,m_{\max,2},b_2) \cap \Tilde{\bags}(k-1, m_{\max,2}, b_2)$ on $m_{\max,2}$ machines identical to the schedule corresponding to $\Tilde{\bags}(k-1, m_{\max,2}, b_2)$ and schedule the bags in $\mathcal{A}(m_{\max,1},b_1,m_{\max,2},b_2) \cap \bags[k,m_{\max,1},b_1]$ on the remaining $m_{\max,1}$ machines identical to the schedule corresponding to $\bags[k,m_{\max,1},b_1]$ on $m_{\max,1}$ machines.
    \end{itemize}
    Notice that in the scenarios where $m \in [1,m_{\max,1}]$
    \begin{align}
        \label{lem:rounding_ineq3}
        l(\mathcal{A}(m_{\max,1},b_1,m_{\max,2},b_2),m)
            & \geq l(\bags[k,m_{\max,1},b_1],m),
    \end{align}
    and that, by \Cref{cor:headgap_empty}, in the scenarios where $m \in (m_{\max,1},k]$
    \begin{align}
        \label{lem:rounding_ineq4}
        l(\mathcal{A}(m_{\max,1},b_1,m_{\max,2},b_2),m)
            & \geq l(\bags(k-1,m_{\max,2},b_2),m).
    \end{align}

    Among the constructed candidate set of bags $\mathcal{A}(m_{\max,1},b_1,m_{\max,2},b_2)$, we store in $\mathcal{C}[k-1, m_{\max}, b]$ the one that maximizes $v_{m_{\max}}(\mathcal{A}(m_{\max,1},b_1,m_{\max,2},b_2))$ together with its corresponding schedule.
    Observe that because each bag in $\bags[k,m_{\max,1},b_1]$ and each bag in $\bags(k-1,m_{\max,2},b_2)$ is either empty or contains only jobs whose processing time is in the same extended interval, it follows that each bag in $\bags[k-1, m_{\max}, b]$ is either empty or contains only jobs whose processing time is in the same extended interval, satisfying the first property of the DP table.

    Let $m_{\max,1}$ be the largest number of machines such that $\opt(\optbags,m_{\max,1}) \in U_{k}$ and let $m_{\max,2}$ be the number of scenarios $m$ for which $\opt(\optbags,m) \in \Tilde{I}^+_{k-1}$, and let $b_1$ be the number of bags in $\optbags$ with jobs from $H_{k}$ and $b_2$ be the number of bags in $\optbags$ with jobs from $J_{k-1}$.
    Thus, by the second property of our DP table, we have that 
    \begin{align}
        \label{lem:rounding_ineq2}
        \sum_{m=1}^{m_{\max,1}} q_m \cdot l(\bags[k, m_{\max,1}, b_1],m)
                & \geq \frac{1}{(1 + \epsilon)\alpha} \sum_{m=1}^{m_{\max,1}} q_m \opt(\optbags,m),
    \end{align}
    Let $m_{\max}' = m_{\max,1} + m_{\max,2}$ and $b' = b_1 + b_2$, and observe that $m_{\max}'$ is the largest number of machines such that $\opt(\optbags,m_{\max}') \in U_{k-1}$ and that $b'$ is the number of bags in $\optbags$ with jobs from $H_{k-1}$.
    We will show that $\mathcal{C}[k-1, m_{\max}', b']$ satisfies the second property of the DP table.
    
    Consider the constructed subproblem from which we obtained the $(1+\epsilon)\alpha$-approximate solution $\Tilde{\bags}(k-1,m_{\max,2},b_2)$.
    Let $\Tilde{\optbags}(k-1, m_{\max,2}, b_2)$ denote an optimal set of bags for this subproblem, let $\Tilde{O}(\Tilde{\optbags}(k-1, m_{\max,2}, b_2), m)$ denote the value of its corresponding optimal schedule for the scenario with $m$ machines.
    Recall that, as a consequence of \Cref{lemma:m_m1_machines_only_jobs_from_J_i}, we assumed that for each scenario $m \in [m_{\max,1},m_{\max,2}]$, there are at least $m - m_{\max,1} + 1$ machines containing only bags with jobs from $J_{k-1}$.
    Therefore, the expected value of an optimal solution to the original instance cannot be strictly greater than the expected value of the optimal solution to the subproblem, i.e., $\sum_{m=1}^{m_{\max,2}} \Tilde{q}_{m} \cdot \Tilde{O}(\Tilde{\optbags}(k-1, m_{\max,2}, b_2), m) \geq \sum_{m=1}^{m_{\max,2}} \Tilde{q}_{m} \cdot \opt(\optbags,m_{\max}'-m_{\max,2}+m)$.
    Hence,
    \begin{align}
        \label{lem:rounding_ineq1}
        \;& \sum_{\mathclap{\substack{m=m_{\max}-m_{\max,2}+1}}}^{m_{\max}'} q_{m} \cdot l(\bags(k-1, m_{\max,2}, b_2),m)\\
             = \;& \sum_{m=1}^{m_{\max,2}} q_{k'-m_{\max,2}+m} \cdot l(\bags(k-1, m_{\max,2}, b_2),m) \nonumber \\
             = \;& Q^{m_{\max}'}_{m_{\max}'-m_{\max,2}+1} \sum_{m=1}^{m_{\max,2}} \Tilde{q}_{m} \cdot l(\bags(k-1, m_{\max,2}, b_2),m) \nonumber \\
             \geq \;& \frac{Q^{m_{\max}'}_{m_{\max}'-m_{\max,2}+1}}{(1+\epsilon) \alpha} \sum_{m=1}^{m_{\max,2}} \Tilde{q}_{m} \cdot \Tilde{O}(\Tilde{\optbags}(k-1, m_{\max,2}, b_2), m) \nonumber \\
            \geq \;& \frac{Q^{m_{\max}'}_{m_{\max}'-m_{\max,2}+1}}{(1+\epsilon) \alpha} \sum_{m=1}^{m_{\max,2}} \Tilde{q}_{m} \cdot \opt(\optbags,m_{\max}'-m_{\max,2}+m) \nonumber \\
              =  \;& \frac{1}{(1+\epsilon) \alpha} \sum_{m=k'-m_{\max,2}+1}^{m_{\max}'} q_{m} \cdot \opt(\optbags,m).
    \end{align}
    Therefore, for $m_{\max}'$ and $b'$ we have that
    \allowdisplaybreaks
    \begin{align*}
        v_{m_{\max}}(\bags[k-1, m_{\max}', b'])
            & \;\;= \max_{m_{\max,1}',b'_1,m_{\max,2},b'_2} \{v_{k}(\mathcal{A}(m_{\max,1}',b'_1,m_{\max,2}',b'_2))\} \\
            & \;\;\geq v_{m_{\max}}(\mathcal{A}(m_{\max,1},b_1,m_{\max,2},b_2)) \\
            & \;\;= \sum_{m = 1}^{m_{\max,1}} q_{m} \cdot l(\mathcal{A}(m_{\max,1},b_1,m_{\max,2},b_2),m) \\&\hspace{30pt} + \sum_{m = m_{\max}'-m_{\max,2}+1}^{m_{\max}'} q_{m} \cdot l(\mathcal{A}(m_{\max,1},b_1,m_{\max,2},b_2),m) \\
            & \stackrel{(\ref{lem:rounding_ineq2}),(\ref{lem:rounding_ineq1})}{\geq} \frac{1}{(1+\epsilon) \alpha} \sum_{m=1}^{m_{\max,1}} q_{m} \cdot \opt(\optbags,m) \\&\hspace{30pt} + \frac{1}{(1+\epsilon) \alpha} \sum_{m=m_{\max}'-m_{\max,2}+1}^{m_{\max}'} q_{m} \cdot \opt(\optbags,m) \\
            & \;\;= \frac{1}{(1+\epsilon)\alpha} \sum_{m=1}^{m_{\max}'} q_{m} \opt(\optbags,m),
    \end{align*}
    and thus $\mathcal{C}[k-1, m_{\max}', b']$ satisfies the second property as required.

    We conclude that we can compute the DP table and obtain an $(1+\epsilon)\alpha$-approximate solution $\mathcal{C}[1, M, M]$ in polynomial time, while losing a factor of at most $1+\mathcal{O}(\epsilon)$ due to our assumptions by \Cref{cor:ingnore_eps_scenarios} and \Cref{lem:bags_restricted_to_intervals}, \Cref{lem:bags_in_interval_dominates_opt_in_interval}, \Cref{lem:headgap_empty}, and~\Cref{lemma:m_m1_machines_only_jobs_from_J_i}. 
\end{proof}

\subsection{Proof of \Cref{lem:sc_root_size}}
\lemscrootsizeassign*
\begin{proof}
Recall that there are at most $3\log_{1+\epsilon}(1/\epsilon)$ many possible values $\ell \in \mathbb{N}$ such that $(1+\epsilon)^\ell \in I_{\kmax}$ or $(1+\epsilon)^\ell \in I_{\kmax-1}$. Therefore, there are $6\log_{1+\epsilon}(1/\epsilon)$ many distinct size estimates $\ell(B)$ for a bag $B\in \optbags[\kmax]\cup\optbags[\kmax-1]$. For each of these we guess the number of bags in time $\bigO (M^{6\log_{1+\epsilon}(1/\epsilon)})$. Furthermore, we have guessed the bags in $\optbags[\kmax]$ and $\optbags[\kmax-1]$ and by definition of $I_1,\dots,I_{\kmax}$ there are at most $1/\epsilon^6$ such bags. Due to our rounding, the jobs in $I_{\kmax}$ and $I_{\kmax -1}$ have at most $6\log_{1+\epsilon}(1/\epsilon)$ many different processing times. Thus, for each bag $B \in \optbags[\kmax]$  we can guess the number of jobs of each processing time assigned to $B$ in time $\bigO (n^{6\log_{1+\epsilon}(1/\epsilon)})$. Doing this for all bags can be done in time $\bigO (n^{\frac{1}{\epsilon^6}6\log_{1+\epsilon}(1/\epsilon)})$. A similar argument holds for guessing the assignment of jobs in $J_{\kmax-1}\cup J_{\kmax-2}$ to bags in $\optbags[\kmax-1]$. The lemma then follows from the fact that $M < n$. 
\end{proof}

\subsection{Proof of \Cref{lem:sc_root_main}}
\lemscrootmain*
\begin{proof}
Let $m \leq m_{\max}^{(K)}$. We start by computing an assignment of bags $\guessbags[\kmax] \cup \hat{\bags}_{\kmax-1}$ and a set of dummy jobs $J_T$ containing $T$ jobs each with processing time $1$ to the $m$ machines. 
This assignment will either give us a value $\ALG(m)$ or assert that our guess was incorrect. 
The assignment looks as follows. We consider every possibility of assigning the bags $\guessbags[\kmax] \cup \hat{\bags}_{\kmax-1}$ to $m$ machines and assign the dummy jobs $J_T$ by simultaneously filling up the least loaded machines to the second lowest machine load until all dummy jobs are assigned.  
 
Observe that each machine can be assigned at most $\bigO_\epsilon(1)$ many jobs from $\guessbags[\kmax] \cup \hat{\bags}_{\kmax-1}$.
Since there are at most $\bigO_\epsilon(1)$ many different size estimates for these bags, the number of different assignments of a subset of these bags to a single machine is at most $\bigO_\epsilon(1)$. 
Therefore, for each of these possible assignments we can guess how many machines receive this assignment in time $n^{\bigO_\epsilon(1)}$. 
The remaining procedure to distribute the dummy jobs $J_T$ can also be done in time $n^{\bigO(1)}$. 
At the end, among all possibilities to assign the bags $\optbags[\kmax] \cup \optbags[\kmax-1]$, we choose the one which maximizes the minimum machine load after also assigning $J_T$. 
If the minimum machine load of the best candidate solution is less than $(1+\epsilon)^{-1}(\frac{1}{\epsilon})^{3K}$, we assert that our guess was incorrect. If the minimum machine load of the best candidate solution is at least $(1+\epsilon)^{-1}(\frac{1}{\epsilon})^{3K}$, we return this value as $\ALG(m)$.

Next, consider a partition $\bags'\cup\{B_{S}, B_{\bar{S}}\}$ of the jobs $\bigcup_{k=1}^{\kmax-2} J_k$ into $M-|\guessbags[\kmax]|-|\guessbags[\kmax-1]|+2$ bags with the properties stated in the lemma and some fixed $m \leq m^{(\kmax)}_{\max}$. 
We will now show that $\opt(\bags'\cup\guessbags[\kmax]\cup\guessbags[\kmax-1],m)\in [(1+\epsilon)^{-5}\ALG(m),(1+\epsilon)\ALG(m))$. 
To this end, fix an optimal assignment of $\bags'\cup\guessbags[\kmax]\cup\guessbags[\kmax-1]$ to $m$ machines; for simplicity, we drop the dependency on the bags and denote this assignment as well as its objective function value by $\opt(m)$, i.e., \[
    \opt(m) := \opt(\bags'\cup\guessbags[\kmax]\cup\guessbags[\kmax-1],m) \, .
\] 

Let $\ALG'$ be the value (and, by overloading notation, also the assignment) obtained by the procedure to compute $\ALG(m)$ when fixing the assignment of the bags $\bags'\cup\guessbags[\kmax]\cup\guessbags[\kmax-1]$ to be the same as in $\opt(m)$.
We will first show that $\opt(m) \in [(1+\epsilon)^{-2}\ALG',(1+\epsilon)\ALG']$.

We first prove that $\opt(m) \leq (1+\epsilon)\ALG'$. 
To this end, we distribute the total processing time $p(\bags')$ of bags $\bags'$ as we distributed the set of dummy jobs $J_T$. 
Since $T$ is an upper bound on $p(\bags')$, each machine receives no more dummy jobs than in $\ALG'$. 
As $\opt(m)$ must assign the processing volume $p(\bags')$ according to the bags $\bags'$, the minimum machine load in $\ALG'$ is at least the minimum machine load in $\opt(m)$; in fact, the former can be viewed as a relaxation of the latter. 
Thus, it follows that \[
    \opt(m) \leq (1+\epsilon)\ALG'\, , 
\]
where the factor $(1+\epsilon)$ is only due to the size estimates of the bags $\bags'\cup\guessbags[\kmax]\cup\guessbags[\kmax-1]$.

Next, we prove that $\opt(m) \geq (1+\epsilon)^{-2}\ALG'$. 
Suppose for the sake of contradiction that $(1+\epsilon)^2\opt(m) < \ALG'$. 
Observe that the total processing time of the dummy jobs available to $\ALG'$ is at most $(1+\epsilon)p(\bags')$. 
Therefore, by reducing the amount of dummy jobs equally on every machine such that we use exactly $|p(\bags')|$ dummy jobs, we get a solution $\ALG'' \geq (1+\epsilon)^{-1}\ALG'$ (conflating again the assignment with its objective function value). 
We now consider the bags $\bags'$ and assign them greedily to the machines such that each machine receives almost the same amount of dummy jobs as in the assignment underlying $\ALG''$. 
Formally, we start with a machine with the smallest number of dummy jobs and greedily add bags from $\bags'$ until the next bag would exceed the number of dummy jobs assigned to this machine in the assignment $\ALG''$. 
Since, for every bag $B \in \bags'$, we have $p(B) \leq \epsilon \opt(m)$ by assumption, this way we find an integer assignment of the bags in $\bags'$ with value at least 
\begin{align*}
    \ALG'' - \eps \opt(m)     & \geq (1+\epsilon)^{-1}\ALG' - \eps \opt(m) > (1+\eps) \opt(m) - \eps \opt(m) \\
    & = \opt(m) \, , 
\end{align*}
which contradicts the optimality of the original assignment. Therefore, we must have that
\[
    \opt(m) \geq (1+\epsilon)^{-2}\ALG' \, . 
\]

It remains to show $\opt(m)\in [(1+\epsilon)^{-5}\ALG(m),(1+\epsilon)\ALG(m)]$. 
Observe that our procedure to compute $\ALG(m)$ ensures $\ALG(m) \geq \ALG'$. 
Thus, by the statements shown above, it follows that 
\[
    \opt(m) \leq (1+\eps) \ALG(m) \,.
\] 

We will now argue that $\ALG(m) \leq (1+\epsilon)^3\ALG'$ which implies that $\opt(m) \geq (1+\epsilon)^{-5}\ALG(m)$ as desired. 
Suppose for the sake of contradiction that $\ALG(m) > (1+\epsilon)^{3}\ALG'$. 
Thus, $\ALG(m) > (1+\epsilon)^{2}\opt(m)$.
We know that $p(\bags') \geq (1+\epsilon)^{-1} T$. Hence, by increasing every bag in $\bags'$ equally such that $p(\bags') = T$, we increase each bag by at most a factor of $(1+\epsilon)$. 
Let $\bags''$ denote these modified bags. 
We can now use $\bags''$ to find an assignment of the bags $\guessbags[\kmax]\cup\guessbags[\kmax-1]$ and $\bags'$: We first assign $\guessbags[\kmax]\cup\guessbags[\kmax-1]$ as in $\ALG(m)$ and then greedily assign the bags $\bags''$ to machines until the next bag would increase the load of the machine beyond its load in $\ALG(m)$. 
We then assign any remaining bags from $\bags''$ arbitrarily. 
Observe that, since each bag has size at most $(1+\epsilon) \epsilon \left(\frac{1}{\epsilon}\right)^{3\kmax} \leq \eps \opt(m)$, this assignment yields an objective function value of at least $\ALG(m) - \eps\opt(m)$. 
Now, replacing bags $\bags''$ with $\bags'$ yields an assignment with objective function value at least 
$(1+\epsilon)^{-1}\ALG(m) - \eps \opt(m) > \opt(m)$, which contradicts the optimality of $\opt(m)$. 
Therefore, \[
\ALG(m) \leq (1+\epsilon)^3\ALG' \leq (1+\eps)^5 \opt(m)      \, ,
\]
using the above shown relationship between $\ALG'$ and $\opt(m)$.
Combining with the above, it follows that $\opt(m)\in [(1+\epsilon)^{-5}\ALG(m),(1+\epsilon)\ALG(m))$ as desired. 
Since the arguments above hold for an arbitrary $m\leq m^{(\kmax)}_{\max}$, this concludes the proof of the lemma.
\end{proof}


\subsection{Proof of \Cref{lem:lmscassginDP}}
\lemscassignDP*
\begin{proof}
Observe that by definition of the currently considered DP cell, we are given values $s_\ell$ for every $\ell \in \mathbb{N}$ such that $(1+\epsilon)^\ell \in I_k$ indicating the number of bags with a size estimate corresponding to $\ell$ and that $s_\ell \leq M$. Furthermore, there are at most $6\log_{1+\epsilon}(1/\epsilon)$ many distinct processing times in the intervals $J_{k}$ and $J_{k-1}$ and a bag in $B\in \optbags[k]$ contains at most $1/\epsilon^6$ jobs from these intervals as otherwise the total size of the bag is larger than $\frac{1}{\epsilon}^{3k + 3}$. Therefore, the number of different assignments of jobs from $J_{k}$ and $J_{k-1}$ to such a bag $B$ is at most
\[
(1/\epsilon^6)^{6\log_{1+\epsilon}(1/\epsilon)}.
\]
We now guess for every $\ell \in \mathbb{N}$ such that $(1+\epsilon)^\ell \in I_k$ how many of the $s_\ell$ bags contain each of the possible different assignments. As there are at most $6\log_{1+\epsilon}(1/\epsilon)$ many values of $\ell$ this can be done in time
\[
M^{(1/\epsilon^6)^{6\log_{1+\epsilon}(1/\epsilon)}\cdot 6\log_{1+\epsilon}(1/\epsilon)}.
\]
The lemma then follows from the fact that $M < n$. 
\end{proof}

\subsection{Proof of \Cref{lem:sc_DP_main}}
\lemscDPmain*
\begin{proof}
We prove this lemma using a similar line of arguments as in the proof of \Cref{lem:sc_root_main}.
We start by describing how to compute the vector $\left(\ALG(m)\right)_{m=m_{\min}^{(k)}}^{m_{\max}^{(k)}}$ or assert that our guess was incorrect. 
To this end, consider some $m\in \{m^{(k)}_{\min}, \ldots, m^{(k)}_{\max}\}$. 
We now compute an assignment of bags $\bags_L \cup \guessbags[k] \cup \hat{\bags}_{k-1}$ and a set of dummy jobs $J_T$ containing $T$ jobs each with processing time $1$ to the $m$ machines. 
This assignment will either give us a value $\ALG(m)$ or assert that our guess was incorrect. 

We first assign the bags $\bags_L$ such that each machine receives at most one such bag. 
Since $p(B) \geq \left(\frac{1}{\epsilon}\right)^{3k+3}$ for each $B \in \bags_L$, the load on each of these machines is sufficiently large. 
Therefore, we can continue with $m-\noBags$ machines and focus on the bags $\guessbags[k] \cup \hat{\bags}_{k-1}$ and the set of dummy jobs $J_T$. 
For the sake of a simpler exposition, we set $m' := m-\noBags$ for the remainder of the proof as this is the ``relevant'' number of machines. 
Now we use the same procedure as in the proof of \Cref{lem:sc_root_main}.
We consider every possibility of assigning the bags $\guessbags[k] \cup \hat{\bags}_{k-1}$ to $m'$ machines and assign the dummy jobs $J_T$ as in \Cref{lem:sc_root_main}, i.e., simultaneously filling up the least loaded machines to the second lowest machine load until all dummy jobs are assigned.  
Observe that each machine can be assigned at most $\bigO_\epsilon(1)$ many jobs from $\guessbags[k] \cup \hat{\bags}_{k-1}$.
Since there are at most $\bigO_\epsilon(1)$ many different size estimates for these bags, the number of different assignments of a subset of these bags to a single machine is at most $\bigO_\epsilon(1)$. 
Therefore, for each of these possible assignments we can guess how many machines receive this assignment in time $n^{\bigO_\epsilon(1)}$. 
The remaining procedure to distribute the dummy jobs $J_T$ can also be done in time $n^{\bigO(1)}$. 
At the end, among all possibilities to assign the bags $\optbags[\kmax] \cup \optbags[\kmax-1]$, we choose the one which maximizes the minimum machine load after also assigning $J_T$. 
If the minimum machine load of the best candidate solution is less than $(1+\epsilon)^{-1}(\frac{1}{\epsilon})^{3k}$, we assert that our guess was incorrect. If the minimum machine load of the best candidate solution is at least $(1+\epsilon)^{-1}(\frac{1}{\epsilon})^{3k}$, we return this value as $\ALG(m)$.

Next, we fix $m$ and prove the desired bound on $\opt(\bags_{L}\cup\hat{\bags}_{k}\cup\hat{\bags}_{k-1}\cup\bags',m)$ assuming that our procedure above computed the vector $\left(\ALG(m)\right)_{m=m_{\min}^{(k)}}^{m_{\max}^{(k)}}$. 
To this end, we drop the dependency on the set of bags and let 
\[
\opt(m) :=  \opt(\bags_{L}\cup\hat{\bags}_{k}\cup\hat{\bags}_{k-1}\cup\bags',m) \,  ;
\]
we also use $\opt(m)$ to denote a fixed assignment of $\bags_{L}\cup\hat{\bags}_{k}\cup\hat{\bags}_{k-1}\cup\bags'$ to $m$ machines attaining this objective function value. 
Similarly, we let $\ALG(m)$ also denote one assignment computed by our procedure attaining an objective function value of $\ALG(m)$. 

Observe that the bags $\bags_L$ are assigned in $\ALG(m)$ as in $\opt(m)$ up to a permutation of machines for any partition of the jobs $J_{1}\cup...\cup J_{\kmax-2}$ into $M-\noBags-|\guessbags[k]|-|\guessbags[k-1]|+1$ denoted by $\bags'\cup\{B_{S'}\}$. 
Therefore, for any such $\bags'\cup\{B_{S'}\}$ satisfying the conditions of the lemma, it suffices to compare the minimum machine load in $\opt(m)$ restricted to the machines without a bag from $\bags_L$ to the minimum machine load in $\ALG(m)$ under the same restriction. 
This is exactly the statement we have already proven in \Cref{lem:sc_root_main}. 
Hence, we do not repeat the proof here. 

It remains to argue that we can indeed identify the best DP cell $\hat\calC$ fitting our guess in polynomial time. 
To this end, recall that there are at most $3\log_{1+\epsilon}(1/\epsilon)$ many possible values $\ell \in \mathbb{N}$ such that $(1+\epsilon)^\ell \in I_{k-1}$. 
Therefore, there are $3\log_{1+\epsilon}(1/\epsilon)$ many distinct size estimates $\ell(B)$ for a bag $B\in \optbags[k-1]$. 
Similar as in the proof of \Cref{lem:lmscassginDP}, the number of different bag sizes in $I_{k-1}$ is bounded by $n^{\bigO_\eps(1)}$ and for a given set of bag sizes, there are at most $n^{\bigO_\eps(1)}$ many assignments of jobs in $J_{k-1}\cup J_{k-2}$ to those bags. 
Hence, there are at most $n^{\bigO_\eps(1)}$ cells $\hat\calC$ that we need to evaluate using the above described algorithm. 
Therefore, in polynomial time, we can select the best DP cell.
\end{proof}

\subsection{Proof of \Cref{thm:santa}}

In order to prove \Cref{thm:santa}, we first prove the correctness before analyzing the running time of the DP described in \Cref{sec:santaclaus}.

To this end, we fix a particular optimal solution $\optbags$. 
We define a special DP cell with respect to $k \in [\kmax-1]$ as the DP cell with first entry $k$ such that all other parameters match the optimal solution. Formally, a DP cell $\mathcal C_k$ is \emph{special with respect to} $k \in [\kmax-1]$ if the following holds
\begin{itemize}
    \item $k$
    \item $ \noBags      = |\optbags[k+1]\cup \dots \cup \optbags[\kmax]|$
    \item $M_{k} = |\optbags[k]|$
    \item For each $m \leq m_{\min}^k - 1$, it holds that $\opt(\optbags,m) \in I_{k+1}\cup \ldots \cup I_{K}$. 
    \item For each $\ell \in \mathbb{N}$ with $(1+\epsilon)^{\ell} \in I_k$,  $s_\ell = |\{B\in \optbags[k]: p(B) \in [(1+\epsilon)^\ell,(1+\epsilon)^{\ell +1}) \}|$.
    \item For each $\ell \in \mathbb{N}$ with $\lceil(1+\epsilon)^{\ell}\rceil \in I_k$, $a_\ell$ counts the jobs with processing time $\lceil(1+\epsilon)^{\ell}\rceil$ assigned to bags in $\optbags[k+1]$.
    \item $S = \sum_{B \in \optbags[k+1]\cup \dots \cup \optbags[\kmax]}\max\{\lceil (1+\epsilon)^\ell\rceil-p^{+}(B),0\}$, where $p^{+}(B)$ is the total processing time of jobs from $J_{k+1} \cup \dots \cup J_{\kmax}$ assigned to $B \in \optbags[k+1] \cup  \ldots \cup \optbags[\kmax]$.    
\end{itemize}

Since the parameters above are well-defined for any choice of $k\in [\kmax-1]$, there exists at least one DP cell $\mathcal{C}$ which is special w.r.t. to $k$. 
Combining this with \Cref{lem:sc_DP_main}, we prove the following lemma.

\begin{lemma}\label{lem:sc_DP_profits}
For any $k\in [\kmax-1]$, the following two statements are true:
\begin{enumerate}[(i)]
    \item The value stored in the corresponding special cell $\mathcal{C}_k$ satisfies
$$
    \mathrm{profit}(\mathcal{C}_k) \geq (1+\epsilon)^{-1}\sum_{m=m^{(k)}_{\min}(\mathcal{C}_k)}^M\opt(\optbags,m) \, .
$$ 
\item There is a partial solution $\bigcup_{k'=1}^{k} \guessbags[k']$ which 
\begin{itemize}
    \item uses a subset of $\bigcup_{k'=1}^{k} J_{k'}$ to fill each bag in $\bigcup_{k'=1}^{k} \guessbags[k']$ to a level of at least $(1+\eps)^{\ell(B) - 1}$,
    \item leaves enough total volume of jobs from $\bigcup_{k'=1}^{k} J_{k'}$ (given by $S$ and $(a_\ell)_{\ell \in L_k}$) necessary to fill larger bags, 
    \item and in combination with the correct number of larger bags (determined by parameter $\noBags$ of the cell) achieves an objective function value of at least $(1+\eps)^{-5} \mathrm{profit}(\calC_{k-1})$.
\end{itemize}
\end{enumerate}

\end{lemma}

\begin{proof}
    Recall that the algorithm underlying \Cref{lem:sc_DP_main} evaluates DP cells by testing all feasible guesses matching the respective cell's parameters and storing the one attaining the best objective function value. Hence, for proving the lemma, it is sufficient to show that for each $k \in [\kmax-1]$ there is a particular guess considered during the evaluation of $\calC_k$ that matches the fixed optimal solution $\optbags$ and that corresponds to an actual achievable solution. 

    We prove the existence of this guess by induction. 
    First, consider $k=1$ and the corresponding special cell $\mathcal{C}_1$.
    Since there is no remaining subproblem, we know that for $m^{(1)}_{\max}(\calC_1)$ the only feasible guess is $M$ and guess the correct value for $m^{(1)}_{\min}(\calC_1)$.
    Further, we set $\guessbags[1] = \optbags[1]$ and assign $J_1$ to $\guessbags[1]$ as to $\optbags[1]$. The remaining parameters are set as defined above for the $\calC_{1}$ implied by $\optbags$. Clearly, this is a valid guess matching the parameters of $\calC_1$, and hence, it is considered when evaluating $\calC_1$. 
    Let $\ALG^*(m)$ denote the objective function value achieved by the procedure of \Cref{lem:sc_DP_main} when considering the above described guess. 
    By \Cref{lem:sc_DP_main}, we have 
    \[
        \mathrm{profit}(\mathcal{C}_1) \geq \sum_{m=m^{(1)}_{\min(\calC_1)}}^M q_m \ALG^*(m) \geq (1+\epsilon)^{-1}\sum_{m=m^{(1)}_{\min(\calC_1)}}^M q_m\opt(\optbags,m) \, .
    \]
    Furthermore, to the correct assignment of jobs from $J_1$ to $\guessbags[1]$, the jobs from $J_1$ not packed in $\guessbags[1]$ or reserved for $\optbags[2]$ contribute enough volume to cover each bag $B \in \bigcup_{k=3}^{\kmax} \optbags[k]$ up to at least $(1+\eps)^{-1} p(B)$ when the jobs from $\bigcup_{k=2}^{\kmax} J_k$ are assigned as in the optimal solution and the jobs from $J_1$ are assigned greedily. By~\Cref{lem:sc_DP_main}, combining this with $M_{2,\ldots,\kmax}$ many large bags achieves a partial objective function value of at least 
    \[
        (1+\eps)^{-5}\sum_{m=m^{(1)}_{\min}(\calC_1)}^M q_m \ALG^*(m) \,.
    \]

    Thus, we can assume that both statements hold for all $k' \leq k-1$ and consider $k \leq \kmax -1$ with special cell $\calC_k$. 
    As for $k=1$, we now define a feasible guess considered during the evaluation of $\calC_k$. 
    To this end, let $\guessbags[k] = \optbags[k]$ and assign $J_k \cup J_{k-1}$ to $\guessbags[k]$ as to $\optbags[k]$. 
    The remaining subproblem is given by $\hat \calC := \calC_{k-1}$, which satisfies \eqref{eq:sc_dp_next} by definition of special cells. 
    By our induction hypothesis and the fact that we combine $\calC_{k-1}$ with $\calC_k$, the jobs from $J_{k-1}$ we have just assigned to $\guessbags[k]$ were not packed in $\guessbags[k-1]$ in the previous step. 
    Further, using again our induction hypothesis, we have enough volume from $\bigcup_{k'=1}^{k-2} J_{k'}$ left to cover the bags in $\guessbags$ up to at least a $(1+\eps)$ fraction of their size. 
    By greedily assigning these jobs, we do not use more volume than the optimal solution did, which implies that those jobs combined with the not yet packed jobs from $J_{k-1}$ have enough volume to cover each bag $B \in \bigcup_{k'=k+1}^{\kmax} \optbags[k']$ up to at least $(1+\eps)^{-1} p(B)$ when the jobs from $\bigcup_{k'=k+1}^{\kmax} J_{k'}$ are assigned as in the optimal solution and the jobs from $J_1$ are assigned greedily.
    For $m \in \{m_{\min}^{(k)}, \dots, m_{\min}^{(k-1)}-1\}$, let $\ALG^*(m)$ denote again the objective function value attained when the procedure from \Cref{lem:sc_DP_main} evaluates the just defined guess, where $m_{\min}^{(k-1)}-1$ is the machine parameter from $\hat \calC = \calC_{k-1}$.  
    As discussed above, we have 
    \[
        \mathrm{profit}(C_k) \geq \sum_{m=m_{\min}^{(k)}(\calC_k)}^{m_{\min}^{(k-1)}(\calC_{k-1})-1}q_m \ALG^{*}(m) + \mathrm{profit}(\mathcal{C}_{k-1}) \, ,
    \]
    Combining with our induction hypothesis 
    \[
        \mathrm{profit}(\mathcal{C}_{k-1}) \geq (1+\epsilon)^{-1}\sum_{m=m^{(k-1)}_{\min}(\calC_{k-1})}^M q_{m}\opt(\optbags,m) 
    \]
    and \Cref{lem:sc_DP_main} for $\calC_k$ this implies 
    \[
        \mathrm{profit}(C_k) \geq (1+\epsilon)^{-1}\sum_{m=m_{\min}^{(k)}(\calC_{k})}^{m^{(k-1)}_{\min}(\calC_{k-1})-1} q_m \opt(\optbags,m) + (1+\epsilon)^{-1} \sum_{m=m^{(k-1)}_{\min}(\calC_{k-1})}^{M}q_m\opt(\optbags,m) \, . 
    \]

    Further, by~\Cref{lem:sc_DP_main}, combining this guess with $M_{k+1,\ldots,\kmax}$ large bags and the already defined bags $\bigcup_{k'= 1}^{k-1} \guessbags[k']$ achieves a partial objective function value of at least 
    \[
        (1+\eps)^{-5}\sum_{m=m^{(k)}_{\min}(\calC_{k})}^{m^{(k-1)}_{\min}(\calC_{k-1})-1} q_m \ALG^*(m) \, ,
    \] 
    Combining this with our induction hypothesis implies that adding $M_{k,\ldots,\kmax}$ large bags to $\bigcup_{k'=1}^k \guessbags[k']$ guarantees a partial objective function value of at least 
    \[
        (1+\eps)^{-5}\sum_{m=m^{(k)}_{\min}\calC_{k}}^{M} q_m \ALG^*(m) \, ,
    \]     
    which concludes the proof of the lemma. 
\end{proof}

To prove the correctness of the DP it now remains to be shown that there is a DP cell $\mathcal{C}_{\kmax-1}$ representing the remaining subproblem with respect to the correct initial guess. 
Let 
$$
    \calC_{\kmax}:=\{m^{(\kmax)}_{\max},\left(s_\ell\right)_{\ell:(1+\epsilon)^\ell \in I_{\kmax-1} \cup I_{\kmax}},\left(a_\ell\right)_{\ell:\lceil(1+\epsilon)^\ell\rceil \in I_{\kmax-1} \cup I_{\kmax}},S\}
$$
be the correct initial guess. Consider the DP cell defined by the following parameters (which by the correctness of the guess is the special DP cell $\mathcal{C}_{k-1}$):
\begin{itemize}
\item $\kmax-1$
\item $M_{\kmax}:= \sum_{\ell:(1+\epsilon)^\ell \in I_{\kmax}} s_\ell$
\item $M_{k} := \sum_{\ell:(1+\epsilon)^\ell \in I_{\kmax}-1} s_\ell$
\item $m^{\kmax-1}_{\min}= m^{(\kmax)}_{\max}+1$
\item $\left(s_\ell\right)_{\ell:(1+\epsilon)^\ell \in I_{\kmax-1}}$
\item $\left(a_\ell\right)_{\ell:\lceil(1+\epsilon)^\ell\rceil \in I_{\kmax-1} \cup I_{\kmax}}$
\item $S$
\end{itemize}

To complete the formal definition of the DP and prove its correctness, it remains to combine the solution stored in cell $\mathcal{C}_{k-1}$ and the partial solution obtained by $\calC_\kmax$ to a global solution for the whole problem.
We will now describe a procedure to combine the two. To this end, let $\mathrm{sol}(\mathcal{C}_{k-1})$ be the solution stored in cell $\mathcal{C}_{k-1}$ Denote by $\hat{\bags}_1,\dots,\hat{\bags}_{\kmax-1}$ the bags of $\mathrm{sol}(\mathcal{C}_{k-1})$. We know that the job-to-bag assignment for bags $\hat{\bags}_1,\dots,\hat{\bags}_{\kmax-1}$ is already complete and the only thing that remains is to assign the remaining jobs from $J_1\cup \dots \cup J_{\kmax-2}$ to the bags $\optbags[\kmax]$. By definition of our DP, we know that the total processing time of these jobs is at least $S$ which is the total processing time necessary to fill the $\optbags[\kmax]$ up to the lower bound of their size estimates, i.e., $S:=\sum_{B\in\optbags[\kmax]}\max\{\left\lceil (1+\epsilon)^{\ell(B)}\right\rceil -p^{+}(B),0\}$.
Let $J^{\mathrm{rest}}$ be the set of remaining jobs in $J_1\cup \dots \cup J_{\kmax-2}$. We assign these jobs greedily to $\optbags[\kmax]$ starting with an arbitrary bag $B \in \optbags[\kmax]$ and adding jobs from $J^{\mathrm{rest}}$ to this bag until 
there is a job $j \in J^{\mathrm{rest}}$ such that adding $j$ would make the total size of the bag
exceed $\lceil(1+\epsilon)^{\ell(B)}\rceil$. 
Then, we consider the next bag instead.
We repeat this procedure for every bag and assign possibly remaining jobs arbitrarily. This procedure runs in polynomial time and guarantees that every bag is filled up to a factor of $(1+\epsilon)^{-2}$ of its required size since for each bag $B \in \optbags[\kmax]$ we have
$$p(B) \geq \lceil(1+\epsilon)^{\ell(B)}\rceil - \max_{j \in J^{\mathrm{rest}}}p_j \geq \lceil(1+\epsilon)^{\ell(B)}\rceil - \epsilon\left( \frac{1}{\epsilon}\right)^{3\kmax} 
\ge  (1+\epsilon)^{-2}\lceil(1+\epsilon)^{\ell(B)}\rceil.$$
We conclude this section by proving why this combined solution yields a $(1+\epsilon)$-approximation to the main problem and why it can be computed in polynomial time. First of all, observe that since $G^*_{\kmax}$ is the correct initial guess and consider the values $\ALG^{G^*_{\kmax}}$, computed by \Cref{lem:sc_root_main} for this guess. Then, by \Cref{lem:sc_root_main} we have that independent of the job-to-bag assignment for bags $\hat{\bags}_1,\dots,\hat{\bags}_{\kmax-1}$, the objective function obtained for any $m \in [1,m^{(\kmax)}_{\max}(G^*_{\kmax})]$ is at least 
$$(1+\epsilon)^{-1}\opt(\optbags,m).$$ 
Furthermore, by \Cref{lem:sc_DP_profits} we know that the solution stored in DP cell $\mathcal{C}_{\kmax-1}$ guarantees that the weighted objective values over the range $[m^{\kmax-1}_{\min}(\mathcal{C}_{\kmax-1},M]$ is at least
$$(1+\epsilon)^{-1}\sum_{m=m^{(\kmax-1)}_{\min}(\mathcal{C}_{\kmax-1})}^M\opt(\optbags,m).$$
Thus, the combined solution yields an expected objective function value of at least
$$(1+\epsilon)^{-1}\sum_{m=1}^M\opt(\optbags,m).$$

Finally, observe that the number of different DP cells is bounded by $n^{O_\epsilon(1)}$ and the procedure used to compute the solution for each DP cell runs in time $n^{O_\epsilon(1)}$. This concludes the proof of \Cref{thm:santa}.

\end{document}